\def\FormatStyle{Journal}  
\def\InHouse{InHouse}

\newif\ifblind
\blindfalse 

\RequirePackage{ifpdf}
\ifx \FormatStyle \InHouse  
  \documentclass[a4paper,12pt,onecolumn]{article}
  \usepackage[margin=1in]{geometry}
  \usepackage{setspace}
  \usepackage{fancyhdr}
  \usepackage{todonotes}
\else 
  \documentclass[sts,preprint]{imsart}
  \usepackage[disable]{todonotes}
\fi

\usepackage[utf8]{inputenc}
\usepackage[cmex10]{amsmath}
\usepackage{amsthm,amssymb,amsfonts}
\usepackage{dsfont} 
\usepackage{IEEEtrantools}
\graphicspath{{./Figures/}}
\usepackage{xspace}
\usepackage{caption}
\usepackage{subcaption}
\usepackage[titletoc,title]{appendix}
\usepackage{booktabs}

\theoremstyle{plain}
\newtheorem{thm}{Theorem}

\newtheorem{corl}{Corollary}
\newtheorem{defn}{Definition}
\newtheorem{lemma}{Lemma}
\newtheorem{prop}{Proposition}

\usepackage{algorithm}
\usepackage{algorithmic}

\newcommand{\mynotesH}[2][]{\todo[backgroundcolor=blue!20!white,inline,
bordercolor=red]{#2}}

\usepackage{natbib} 
\usepackage{hyperref} 
\definecolor{Pcolor}{rgb}{0.20,0.50,0.40}
\hypersetup{colorlinks=true, citecolor = Pcolor}

\newcommand{\E}{\mathbb{E}}
\newcommand{\Prob}{\mathbb{P}}

\newcommand{\Var}{\text{Var}}
\newcommand{\bbR}{\mathbb{R}}

\newcommand{\I}{\mathbf{I}}

\newcommand{\sP}{\mathcal{P}}

\newcommand{\X}{\mathbf{X}}
\newcommand{\x}{\mathbf{x}}

\newcommand{\Rs}{R$^2$\xspace}
\newcommand{\RAIp}{RAI$^+$\xspace}
\newcommand{\mRs}{\text{\Rs}} 

\DeclareMathOperator*{\argmax}{arg\,max}

\begin{document}

\begin{frontmatter}

\title{Fitting High-Dimensional Interaction Models with Error Control}
\runtitle{High-Dimensional Interaction Models}

\begin{aug}
  \author{Kory D. Johnson\ead[label=e1]{kory.johnson@wu.ac.at}},
  \author{Robert A. Stine\ead[label=e2]{stine@wharton.upenn.edu}},
  \and
  \author{Dean P. Foster\ead[label=e3]{dean@foster.net}}%

  \runauthor{Johnson et al.}

  \affiliation{Wirtschaftsuniversit{\"a}t Wien and Amazon.com NYC}

  \address{Kory D. Johnson Welthandelsplatz 1 1020 Vienna, Austria\printead{e1}}
  \address{Robert A. Stine and Dean P. Foster 950 6th Avenue, New York, NY, 
USA, \printead{e2,e3}}
\end{aug}

\begin{abstract}
There is a renewed interest in polynomial regression in the form of identifying 
influential interactions between features. In many settings, this takes place in 
a high-dimensional model, making the number of interactions unwieldy or 
computationally infeasible. Furthermore, it is difficult to analyze such spaces 
directly as they are often highly correlated. Standard feature selection issues 
remain such as how to determine a final model which generalizes well. This paper 
solves these problems with a sequential algorithm called Revisiting 
Alpha-Investing (RAI). RAI is motivated by the principle of marginality and 
searches the feature-space of higher-order interactions by greedily building 
upon lower-order terms. RAI controls a notion of false rejections and comes with 
a performance guarantee relative to the best-subset model. This ensures that 
signal is identified while providing a valid stopping criterion to prevent 
over-selection. We apply RAI in a novel setting over a family of regressions in 
order to select gene-specific interaction models for differential expression 
profiling.
\end{abstract}

\begin{keyword}
\kwd{Feature Selection}
\kwd{Alpha-Investing}
\kwd{Stepwise Regression}
\kwd{Multiple Comparisons}
\kwd{Interaction}
\end{keyword}

\end{frontmatter}




\mynotesH{Bob: 
Bottom line: Impressive results. Repackage as “Fitting high dimensional 
interaction models” (or something like that), lead with the strong simulated 
examples (and hopefully a real one chosen from biology), then add theory as 
needed to justify what you are doing.}

\mynotesH{Go through entire document and make claims asymptotic. Check for all 
use of gaussian or normal etc}

\listoftodos

\section{Introduction}

We study the problem of selecting predictive features from a large feature
space. Our data consists of $n$ i.i.d. observations of (response, feature) 
sets, where each observation has $m$ associated features:
$(y_i,x_{i1},\ldots,x_{im}) \sim \sP$. As we are interested in models 
using not only these $m$ raw features 
(main effects) but also polynomials of the features (interaction 
effects), we consider an expanded feature space of dimension $d$ 
which is created by extending our data set as 
($y_i,x_{i1},\ldots,x_{im},x_{i(m+1)},\ldots,x_{id})$. Each $x_{ik}$ for $k \in 
\{m+1,\ldots,d\}$ is a product of raw features:
\begin{IEEEeqnarray*}{rCl}
 x_{ik} & = & \prod_{j=1}^m x_{ij}^{\xi_{kj}},
\end{IEEEeqnarray*}
where $\xi_{kj}$ is the power of feature $j$ in the interaction $k$. In 
general, we do not place a constraint on either the value of $\xi_{kj}$ or the 
number of features $j$ that are allowed in the interaction, as these 
characteristics will be adaptively determined by our procedure. The complete 
feature space, however, is best thought of as all interactions which can be 
created with the restriction that $\sum_{j=1}^m \xi_{kj} \leq \kappa$, for some
$\kappa \in \mathbb{N}$.

Observations are collected into matrices and the following
model is assumed for our data
\begin{IEEEeqnarray}{rCl+rCl}
 Y & = & \mu(\X)+ \epsilon & \epsilon & \sim & N_n(\mathbf{0},\sigma^2 \I_n)
 \label{eqn:model}
\end{IEEEeqnarray}
where $\X \in \bbR^{n \times m}$ and $Y \in \bbR^n$. We will produce 
linear estimates of the true conditional mean $\mu(\X)$ of the form $\hat Y = 
\X\hat\beta$ for $\hat\beta\in\bbR^d$. As we fit linear models which always 
include an intercept, we assume 
without loss of generality that our raw features are mean zero, $\sum_i y_i = 
\sum_i x_{ij} = 0$, $\forall j \leq m$, and normalized such that $\Var(Y) = 
\Var(X_j) = 1$, $\forall j \leq m$.  While we fit linear models and 
conduct hypothesis tests on estimated coefficients, 
we do not assume that we are testing within 
the true linear model or that one even exists. Instead, we view null hypotheses 
in terms of projections of the true mean on the observed data as in 
\citet{Abadie14}, also known as 
the ``$\X$-conditional parameter'' in \citet{Buja+18}. This is necessary as we 
will be performing hypothesis tests multiple times in many different models.
More details are given in Section \ref{sec:notation}.

Generating good predictions requires 
identifying a comparatively small subset of predictive features. 
In modern applications, $d$ is often exceedingly large, which makes 
the selection of an appropriate subset of these features essential.
The model selection problem is to minimize the error sum of squares
\begin{IEEEeqnarray}{rCl}
 \text{ESS}(\hat Y) = \|Y - \hat Y\|^2_2 = \sum_{i=1}^n (Y_i - \hat Y_i)^2 
\label{eqn:ESS}
\end{IEEEeqnarray}
while restricting the number of nonzero coefficients:
\begin{IEEEeqnarray}{c't'rCl}
 \min_{\delta \in \bbR^d} \text{ESS}(\X\delta) & s.t. &
\|\delta\|_{l_0} & = & \sum_{i=1}^d I_{\{\delta_{i}\neq0\}} \leq K,
\label{eqn:sparse-reg}
\end{IEEEeqnarray}
where the number of nonzero coefficients, $K$, is the desired sparsity. Note
that we are not assuming that a sparse true model exists, merely asking for a
sparse approximation. In the statistics literature, the model selection
problem (\ref{eqn:sparse-reg}) is more commonly posed as a penalized regression:
\begin{IEEEeqnarray}{rCl}
  \label{eqn:penreg}
\hat\beta_{0,\lambda} & = & \text{argmin}_{\delta}\left\{\text{ESS}(\X \delta) +
  \lambda\|\delta\|_{l_0} \right\}
\end{IEEEeqnarray}
where $\lambda \ge 0$ is a constant. The classical hard thresholding algorithms
$C_p$ \citep{Mal73}, AIC \citep{Aka74}, BIC \citep{Schwarz78}, and RIC
\citep{FosterG94} vary $\lambda$. The solution to (\ref{eqn:penreg})
is the least-squares estimator on an optimal subset of features. Let $M
\subset \{1,\ldots,d\}$ indicate the coordinates of a given model so that
$\X_M$ is the corresponding submatrix of the data. If $M_\lambda^*$ is the
optimal set of features for a given $\lambda$, then $\hat\beta_{0,\lambda}^* =
\X_{M_\lambda^*}^\dag Y$, where $\X_{M_\lambda^*}^\dag$ is the pseudoinverse of 
$\X_{M_\lambda^*}$.

Given the combinatorial nature of the constraint, solving (\ref{eqn:sparse-reg})
quickly becomes infeasible as $d$ increases and is NP-hard in general
\citep{Nata95}. Forward stepwise is the greedy approximation to the
solution of (\ref{eqn:sparse-reg}). Let $M_i$ be the features in the forward
stepwise model after step $i$ and note that the size of the model is $|M_i| =
i$. The algorithm is initialized with $M_0 = \emptyset$ and iteratively adds
the variable which yields the  largest reduction in ESS. Hence, 
$M_{i+1} = \{M_i\cup j\}$ where
\[j = \argmax_{l \in \{1,\ldots,d\}\backslash M_i}
\text{ESS}(\X_{M_i\cup l} \hat\beta_{M_i\cup l}^\text{LS}).\]
After the first feature is selected, subsequent models are built having fixed
that feature in the model. $M_1$ is the optimal size-1 model, but $M_i$ for
$i\geq 2$ is not guaranteed to be optimal, because $M_i$ is forced to include
the features identified at previous steps. 

To illustrate our method, consider the  concrete compressive strength data from 
the UCI machine learning repository \citep{concrete-data}. It is important to 
note that our goal is to identify \emph{interactions} in contexts in which 
interpreting such features is desirable. Precise examples of this sort are 
provided in Section \ref{sec:poly}. For now, our goal is to demonstrate the 
ability to search such a space to find signal over competing methods. This data 
set is used because the response, compressive strength, is described as a 
``highly nonlinear function of age and ingredients'' such as cement, fly ash, 
water, superplasticizer, etc. It is also useful since it has approximately 1000 
observations and only 8 features. A small number of features is needed so that a 
large, higher-order interaction space can be generated in order to see how 
standard feature selection algorithms perform. All interactions up to fourth 
order are provided to competitor algorithms, in which case there are 1,124 
features.

This paper presents an algorithm, Revisiting-Alpha Investing (RAI), that is able 
to adaptively search a high-dimensional interaction space by building 
interactions from previously selected components. This is done 
while controlling a notion of false rejections. While the details are presented 
in Section \ref{sec:rai}, the idea of RAI is easy to describe: RAI is provided 
the raw features, and conducts an approximate version of stepwise regression on 
these features. Instead of selecting the feature which is \emph{most} 
significant, it merely includes features which are significant enough to pass an 
appropriate statistical test. When a feature is included in the model, all 
possible interactions with \emph{previously selected features} are added to the 
searchable feature space. Stepwise selection then continues on this expanded 
feature space. Furthermore, it is often useful to search for 
interactions of transformed features, e.g. after taking the square-root. This 
allows a fourth-degree polynomial among the new features to be merely a 
second-degree polynomial among the original features. We compare RAI to forward 
stepwise, backward stepwise, Lasso, and FOBA \citep{Zhang08}. 
FOBA stands for ``forward-backward'' and is an algorithm with both sequentially 
adds and deletes features. In this data example, we also use our implementation 
of stepwise as other implementations were not able to cope with even problems 
of this size.

To estimate out-of-sample performance, we create 10 independent
splits of the data into training and test sets, where 25\% of the data is used 
for testing.  We select the parameters of the models using using 5-fold 
cross-validation using the training data and measure performance of the 
CV-selected model on the test set. We compare models using the predictive 
root-mean-squared error (RMSE) on the test set and average model size. Each row 
in Table \ref{tab:concrete} indicates the explanatory features that the 
algorithms are provided. For example, the first row shows the performance 
results when all algorithms are only given raw features to estimate main 
effects, while in subsequent rows the competitor algorithms are given all 
second order interactions etc. We only consider models up to fourth-order 
interactions due to the computational limitations of competitor algorithms. 
This also motivates the desire for a stepwise solution, as other standard 
algorithms are unable to select satisfactory models from interaction spaces.

\mynotesH{plot of MSE on size: get same MSE need more features}
\begin{table}[ht]
\centering
 \caption{Concrete Compression Strength Results.}
 \label{tab:concrete}
\begin{tabular}{rlrrrrr}
  \toprule
 Set & Statistic & foba & lasso & lm & rai & raiStep \\ 
  \midrule
  $\X$ & RMSE& 10.05 & 10.05 & 10.08 & \bf{7.70} & 10.09 \\ 
  $p=11$ & Model Size& 7.30 & 9.40 & 11.00 & \bf{18.60} & 5.90 \\[1ex] 
  $\X^2$ & RMSE& 7.66 & 7.57 & 7.47 &  & 8.00 \\ 
  $p=74$ & Model Size& 40.80 & 66.10 & 74.00 &  & 12.30 \\[1ex] 
  $\X^3$ & RMSE& \bf{6.54} & 8.09e9 & 16.81 &  & 7.59 \\ 
  $p=326$ & Model Size& \bf{59.80} & 288.20 & 326.00 &  & 16.20 \\[1ex] 
  $\X^4$ & RMSE& 19.19 & 7.55e20 & 1.73e6 &  & \bf{7.50} \\ 
  $p=1124$ & Model Size& 86.30 & 562.30 & 1124.00 &  & \bf{20.30} \\ 
   \bottomrule
\end{tabular}
\end{table}

There are several important points in Table \ref{tab:concrete}. As expected, RAI
is superior to other feature selection methods when only considering
marginal features, as the other algorithms do not extend the feature space to 
consider interactions; however, we can adjust for the information differences by
giving the competitor algorithms a richer feature space. When given second 
order interactions, the other algorithms have effectively the same 
performance as RAI but select far more features. Furthermore, as the 
feature space gets more complex, the inherent difficulty becomes apparent 
as the RMSE of the lasso and the full linear model explode. FOBA actually 
continues to improve for third order interactions before worsening, though it 
still selects many more features. Unfortunately, the computational 
complexity of FOBA is even worse than classical stepwise, so we are 
unable to use it in larger examples considered later. Lastly, it is worth 
noting that our implementation of forward stepwise has a stopping criterion 
that prevents it from over-selecting: the performance does not change 
drastically as we go beyond second order interactions. Furthermore, the final 
model selected using RAI on raw features behaves very similarly to the stepwise 
model selected from all forth order interactions. As such, RAI is adequately 
searching this space in a stepwise fashion. 

The remainder of the paper is organized as follows. Section \ref{sec:notation} 
provides the notation and definitions necessary to state our main results. 
Section \ref{sec:approx-step} introduces our threshold approximations via the 
conceptually simpler Revisiting Holm procedure. This section begins with a 
motivating example to demonstrate the differences between post-selection 
inference and using inference to perform model selection. Section \ref{sec:rai} 
provides a more complex threshold procedure which overcomes the shortcoming of 
Revisiting Holm and has performance guarantees. The benefits and validity of 
these procedures are demonstrated via simulation in Section \ref{sec:compare}. 
The interaction-search procedure described above is explained in Section 
\ref{sec:poly}, which includes a novel methodology differential 
expression profiling. Section \ref{sec:discussion} concludes.


\section{Notation and Results}
\label{sec:notation}

We use notation from the multiple comparisons literature given its connection 
to 
our solution. Consider $l$ null hypotheses, $H_{[l]}$: $H_1,\ldots,H_l$, and 
their associated p-values, $p_{[l]}$: $p_1,\ldots,p_l$. The notation $[l] = 
\{1,\ldots,l\}$ will be used frequently for any $l \in \mathbb{N}$. The 
hypotheses are \emph{sequential}, meaning they arrive in a sequence and one 
must 
decide whether or not to reject test $H_i$ before knowing any information about 
$H_{i+1}$. In general, the order of hypotheses can depend on the sequence of 
rejections, but this is not the focus of the current discussion so is ignored. 
As sorting p-values is frequently required, we use standard order statistic 
notation: $p_{(1)} \leq p_{(2)} \leq \ldots \leq p_{(l)}$. Note that this same 
ordering extends to both hypotheses and features: $H_{(k)}$ and $X_{(k)}$ are 
the hypothesis and feature corresponding to $p_{(k)}$, respectively.

As we analyze forward stepwise, each null hypothesis assumes that the fit of 
$M$ is not improved by adding a feature $\X_j$:
\begin{IEEEeqnarray*}{rCl}
 H_i^{M,j}: P_M\mu(\X) & = & P_{M \cup j}\mu(\X),\IEEEyesnumber\label{eqn:null}
\end{IEEEeqnarray*} 
where $M$ is a set of indices specifying the columns of $\X$ in the current 
model, $P_{M'} = \X_{M'}\X_{M'}^\dag$ is the projection matrix onto the 
column-span of $\X_{M'}$, and $\X_{M'}^\dag$ is the pseudo-inverse of $\X_{M'}$ 
for any $M' \subset [m]$.  Similarly, $P_{M'^\perp} = \I-P_{M'}$  is the 
projection onto the space orthogonal to the column-space of $\X_{M'}$. In 
equation (\ref{eqn:null}), the null hypothesis depends on the data $\X$. 

The  sub- and super-scripts in equation (\ref{eqn:null}) are unrelated. The 
subscript 
$i$ merely counts the index of the test whereas the superscripts $M$ and $j$ 
determine the test. Sometimes scripts will be dropped when only the index of the 
test, $i$, or the identity of the test, $\{M,j\}$, is required (these 
conventions also apply to $R$ and $V$ defined below).  Note that order statistic 
notation extends here as well. Given a set of hypotheses 
$H^{M,1},\ldots,H^{M,l}$, their corresponding sorted p-values are $p^{M,(1)} 
\leq p^{M,(2)} \leq \ldots \leq p^{M,(l)}$. The comparisons here are legitimate 
because all $l$ tests are conducted within the same base model $M$.

Note that our null hypothesis is equivalent to testing the coefficient on $\X_j$ 
in the model $M' = M \cup j$. Additional care must be taken with these tests 
as we are not assuming to be testing in the correct model.
\citet{Abadie14} provide a standard error estimate that is robust to model 
mis-specification for our $\X$-conditional parameter. In practice, we 
use a more conservative estimator for computational efficiency as explained in 
Section \ref{sec:rai}. Given such an estimate, $\hat \sigma$, the classical 
test statistic in this setting is 
\[z^{[M,j]} = \frac{Y'P_{M^\perp}\X_j}{\hat\sigma\sqrt{\X_jP_{M^\perp}\X_j}}.\]

As we are using classical p-values, one may wonder at the connections to 
classical rules to stop forward stepwise such as F-to-enter or F-to-delete. 
It was previously demonstrated that such tests do not control
any robust statistical quantity, because attempting to test the addition of new
features uses non-standard and complex distributions \citep{DraperGK71, 
PopeW72}. This critiques are not relevant to our setting, because we are using 
this aforementioned, model-dependent target and modifying the 
stepwise routine to account for selection.

Define the statistic $R_i^{M,j}=1$ if $H_i^{M,j}$ is rejected and $R_i^{M,j}=0$ 
if not. Similarly, let $V_{i,\mu}^{M,j} = 1$ if $R_i^{M,j}=1$ is a 
false rejection ($H_i^{M,j}$ is true) and $V_{i,\mu}^{M,j} = 0$ if 
not. The dependence of $V_{i,\mu}^{M,j}$ on $\mu$ indicates that it is an 
unobserved random variable which depends on the unknown mean. For 
simplicity, all future uses of $V_{i,\mu}^{M,j}$ suppresses this  notation. 
Define
\begin{IEEEeqnarray*}{rCl't'rCl}
 R(l) & = & \sum_{i=1}^l R_i & \text{ and} &
 V(l) & = & \sum_{i=1}^l V_i,
\end{IEEEeqnarray*}
to be the total number of rejections and false rejections in the $l$ 
sequential tests, respectively.


Our method controls the marginal false discovery rate (mFDR) 
which is similar to the more common false discovery rate (FDR):
\begin{defn}[Measures of the Proportion of False Discoveries]
 \label{defn:false_discoveries}
\begin{IEEEeqnarray*}{rCl}
  \text{mFDR}(l) & = & \frac{\E(V(l))}{\E(R(l))+1}\\
  \text{FDR}(l) & = &  \E\left(\frac{V(l)}{R(l)}\right)\text{, where
}\frac{0}{0} = 0.
\end{IEEEeqnarray*}
\end{defn}
In some respects, FDR is preferable to mFDR because it controls a property of a
\emph{realized} distribution. While not observed, the ratio $V(l)/R(l)$ is the
realized proportion of false rejections. FDR controls the expectation of this 
quantity. In contrast, mFDR is
not a property of the distribution of $V(l)/R(l)$. That being said, FDR and mFDR
behave nearly identically in practice, and mFDR yields a powerful and flexible 
martingale \citep{FosterS08}. This martingale provides the basis for proofs of 
mFDR control in a variety of situations.

Our first contribution is an elucidation of the effects that must be considered
when using hypothesis testing for model selection. Standard inference tools are
invalid due to two selection effects: the \emph{ranking} effect and the
\emph{testing} effect. The ranking effect is the result of testing hypotheses
that are suggested by the data and the testing effect is the result of
only conducting future tests if previous tests have been rejected. The impacts
of these effects are explained via example in Section \ref{sec:inf4ms}. 

In Sections \ref{sec:approx-step} and \ref{sec:rai}, we demonstrate that the 
sequential testing approach to multiple comparisons yields an approximate 
forward stepwise algorithm that controls for the selection effects.  We provide 
three procedures of increasing complexity: Revisiting-Holm (RH), 
Revisiting-Alpha-Investing (RAI), and \RAIp. RH is introduced in Section 
\ref{sec:approx-step} purely as a conceptual tool to aid understanding. RAI is 
provided in Section \ref{sec:rai} to avoid serious pitfalls of RH, and \RAIp 
provides minor adjustments to the procedure to enjoy stronger theoretical 
guarantees. In practice, RAI and \RAIp behave essentially identically.

All of the procedures are threshold approximations to stepwise regression.
At each step, forward stepwise computes and sorts the sequential p-values for 
adding each of the remaining $m'$ features to the current model $M$, $p^{M,(1)} 
\leq \ldots \leq p^{M,(m')}$,  and selects the  feature with the minimum 
p-value. Instead of performing a full sort, threshold approximations use a set 
of increasing rejection thresholds, and hypotheses are rejected when their 
p-value falls below a threshold. A feature merely needs to be significant 
\emph{enough}, not necessarily the \emph{most} significant. The initial 
rejection threshold conducts a strict test for which only highly significant 
features are added to the model. Subsequent thresholds perform less stringent 
tests. As such, the final model is built from a series of approximately greedy 
choices.

While stepwise-regression is a traditionally slow procedure, RAI and \RAIp 
achieve a dramatic increase in speed by conducting individual, sequential 
tests and taking advantage of Variance Inflation 
Factor Regression \citep{LinFU11}. If the final model size is of smaller order 
than $\min(n,m)$, the computational complexity of RAI grows at $O(nm\log(n))$. 
Using the full data requires computing $\X'y$, which takes $O(nm)$ time. 
Therefore, RAI merely adds a log factor to perform valid model selection.

Our first result shows that sequential, classical tests are conditionally of 
level $\alpha$ when the null hypotheses are augmented for previously failed 
tests. For example, if we fail to reject $H_i^{M,j}: P_M\mu(\X) = P_{M \cup 
j}\mu(\X)$, then the next test is under the null-hypothesis, $H_{i+1}^{M,j'}: 
P_M\mu(\X) = P_{M \cup j}\mu(\X) = P_{M \cup j'}\mu(\X)$. The alternative in 
this case is still the same, ie, that $P_M\mu(\X) \ne P_{M \cup j'}\mu(\X)$. 
This does not yield the most powerful test of this null hypothesis, but it is 
consistent with our use of a rejection of it. While such augmentation may not 
generally be palatable, it is natural in the sequential 
setting, only concerns the set of failed hypothesis tests, and allows us to 
leverage the Gaussian Correlation Inequality \citep{Royen14,LatM15} in the 
proof of the following theorem. 
\begin{thm}
 \label{thm:valid-tests}
 Using an estimate of $\hat\sigma$ which is robust to model 
misspecification, the sequential testing methods RAI and \RAIp ensure
\[\E(V_i|R_1,\ldots,R_{i-1}) \leq \alpha_i\] 
for every $i < n$ under the augmented null-hypothesis.
\end{thm}

Note that this holds regardless of the correlation structure between the columns 
of $\X$ due to normality. The constraint on the maximum number of tests $i$ is 
due to augmenting the null hypothesis for all failed tests. More information on 
this construction and its implications are in Appendix \ref{app:valid-tests}. 
With Theorem \ref{thm:valid-tests} in hand, all of our procedures control mFDR 
as they are alpha-investing rules \citep{FosterS08}. The algorithms are 
presented independently of alpha-investing so that the algorithm and proof 
method are not conflated.
\begin{corl}
 \label{cor:mfdr-control}
 RAI and RAI$^+$ control mFDR:
  \begin{IEEEeqnarray*}{rCl}
  \frac{\E(V(m))}{\E(R(m))+1} \leq \alpha.
 \end{IEEEeqnarray*}
\end{corl}

Theorem \ref{thm:valid-tests} and Corollary \ref{cor:mfdr-control} are in some 
ways similar to other post-selection inference methods but the perspective is 
very different. \citet{Berk+13} adjust for the potentially adversarial 
instance in which a final model is selected to influence the test-statistic of 
one of the features. Selective inference \citep{FithST15,Taylor+14} control the 
selective type-I error, which conditions on the event that the model $M$ and 
null hypothesis $H_0$ were selected for testing. 

We do not control this selective error rate, because the hypothesis we are 
testing is not actually \emph{intentionally selected}, so to speak, based on 
some criteria. We do not test a feature in a given model because it has the 
minimum p-value, for example. Instead, we produce an algorithm which 
agnostically generates a set of rejections with the desired properties while 
paying sufficient alpha-wealth for their discovery. In this 
way, we need not condition on the question being asked as in selective 
inference, but on the set of answers, ie rejections, that led us to the current 
test. This subtle distinction allows us to leverage classical hypothesis 
tests. In an important sense, we are not doing the ``data snooping'' that 
invalidates these tests.

A second type of result concerns the approximation guarantee between \RAIp and 
stepwise regression. While we are not guaranteed to select the same model as 
forward stepwise, we are able to guarantee that the in-sample fit we achieve is 
comparable to that of forward stepwise and the ideal model.  Our measure of 
model fit for a set of features $\X_M$ is the coefficient of determination, 
\Rs, 
defined as \[\mRs(M) = 1 - \frac{ESS(\X_M \hat\beta_M)}{ESS(\bar Y)} \] where 
$\bar Y$ is the constant vector of the mean response and $\hat\beta_M$ is the 
least squares estimate for predicting $Y$ from $\X_M$. Corollary 
\ref{cor:mfdr-control} coupled 
with Theorem \ref{thm:performance} below demonstrate that \RAIp identifies 
signal (measured in-sample), but controls for over-fitting by not making too 
many incorrect selections.

In the following theorem, for $r \in (0,1)$, $r^s$ is the threshold for 
improvement in \Rs used on the $s$th testing pass. For example, if $r=1/2$, 
then 
the first testing pass searches for features which increase \Rs by 1/2, while 
the second searches for those yielding an increase of 1/4. This yields a 
geometrically decreasing sequence of bounds. The translation between these \Rs 
thresholds to p-value thresholds is given in the proof of the theorem. The 
theorem compares the performance of the selected set of $l$ features, $M_l$, 
and 
the best performing set of $k$ features, $M_k^*$. The term $\gamma$ is similar 
in spirit to the restricted-eigenvalue condition \citep{RasWY10}, but is less 
restrictive and tailored to our setting. It is also lower bounded by a sparse 
eigenvalue. A full discussion is provided in Section \ref{sec:rai}. As will be 
seen in the proof, most instances of the procedure use a bound that is much 
better than $\gamma$, as control with respect to models of this size is only 
required in cases when all features in the final model are rejected in 
a single pass. This only occurs in specially created scenarios.

\begin{thm}
 \label{thm:performance}
 RAI$^+$ selects a set of features $M_l$ of size $l$ such that
 \begin{IEEEeqnarray*}{rCl}
  \mRs(M_l) & \geq & (1-e^{l/c})\mRs(M_k^*)
 \end{IEEEeqnarray*}
 where $c = \left( \frac{\iota+k}{\gamma r} - \iota \right)$ and $\iota$ is the 
maximum number of features rejected in a testing pass. \end{thm} As $r\to1$, 
$\iota\to 0$, and \RAIp converges to stepwise regression. In this case, Theorem 
\ref{thm:performance} recovers the usual bound for stepwise regression 
$\mRs(S_l) \geq (1-e^{l\gamma/k})\mRs(S_k^*)$, which is also proven in the 
appendix.

If the approximation guarantee given above is insufficient and using the exact 
forward stepwise path is desired, \citet{GSell+15} provide a method for turning 
a sequence of valid, independent, post-selection p-values into a model 
selection procedure. While these p-values cannot be validly used to 
select a model \citep{BrownJ16comm}, they derive a procedure in this case 
called  ForwardStop, which can be seen as a limiting case of RAI. 
This produces a final procedure, Stepwise-RAI (S-RAI), for use in small models 
and also demonstrates why ForwardStop has low power. Furthermore, S-RAI is able 
to perform valid model selection using traditional, stepwise p-values.

The sequential testing framework of RAI allows the order of tested hypotheses to
be changed as the result of previous tests. This allows for directed searches
in data base queries or identifying interactions. Section
\ref{sec:poly} leverages this flexibility to greedily search
high-order interactions spaces. Such directed search does not invalidate 
Theorem \ref{thm:valid-tests} or Corollary \ref{cor:mfdr-control} as future 
tests or orthogonal to previous rejections; however, the performance guarantee 
needs to be modified and is presented in in Section 
\ref{sec:poly}. We provide simulations and real data examples to 
demonstrate the success of our method.

\section{Approximating Stepwise Regression}
\label{sec:approx-step}

To motivate the construction of our revisiting procedure, it is instructive to 
consider the problem of stopping stepwise regression in a simple data setting. 
Consider the prostate cancer data used to motivate the 
inference methods of \citet{Taylor+14}. The data set has 67 observations of 
8 
explanatory variables which will be used to predict the log PSA level of men 
who 
had surgery for prostate cancer. The traditional use of stepwise regression is 
summarized in Table \ref{tab:prostate}. Each step of the procedure adds a 
feature to the model and assigns a p-value measuring the reduction in ESS using 
an F-test with independent $\hat\sigma$. The classical p-value does not take 
into account the fact that the hypothesis being tested is suggested by the data 
and algorithm. The second column of p-values in Table \ref{tab:prostate} are 
from \citet{Taylor+14} and adjust for selecting features using forward stepwise.

\begin{table}
 \centering
 \caption{Stepwise Regression: Prostate Cancer Data}
 \label{tab:prostate}
 \begin{tabular}{cccc}
  Step & Parameter & Stepwise p-value & Adjusted p-value\\
  \hline 1 & lcavol & 0.0000 & 0.000\\
  2 & lweight & 0.0003 & 0.006\\
  3 & svi & 0.0424 & 0.425\\
  4 & lbph & 0.0468 & 0.168\\
  5 & ppg45 &  0.2304 & 0.423\\
  6 & lcp & 0.0878 & 0.273\\
  7 & age & 0.1459 & 0.059\\
  8 & gleason & 0.8839& 0.156\\
 \end{tabular}
\end{table}

Our goal is to use the stepwise p-values in Table \ref{tab:prostate} to 
determine when to stop forward stepwise. For example, if it is claimed that the 
first 4 steps are significant but the 5th is not, the selected model will 
include lcavol, lweight, svi, and lbph. Such claims should be made solely on 
the basis of the p-values. We present a framework in which the traditional, 
stepwise p-values can be used to select a model. The p-values can provide 
preferable rejection regions to the adjusted tests of \citet{Taylor+14}
\citep{BrownJ16comm}. We return to the adjusted p-values in Section 
\ref{sec:exactFS} when comparing their use to one of our procedures.  The 
difficulties in post-selection inference are accounted for within the testing 
framework instead of the p-value computation. This requires modifying the 
forward stepwise procedure but Theorem \ref{thm:performance} guarantees we 
select a model which performs similarly.

\subsection{Inference for Model Selection}
\label{sec:inf4ms}

Attempting to use inference \emph{for} model selection poses significantly 
different challenges than merely performing inference \emph{after} a model is 
selected. Inference will be conducted multiple times based on the result of 
previous inferential claims. A simple, orthogonal example distinguishes the 
different challenges posed by conducting inference multiple times. 
This separates questions about the statistical validity of 
repeated testing during forward stepwise from its ability to approximate the 
sparse regression problem 
(\ref{eqn:sparse-reg}): 
the model identified at the $k$th step of forward stepwise exactly solves 
(\ref{eqn:sparse-reg}) under orthogonality. That being said, constructing valid 
tests is still non-trivial due to the selection inherent in forward stepwise 
and repeated testing.

Suppose the data contain 10 orthogonal, explanatory features with true 
parameters $\beta_1=\ldots=\beta_{10}=0$ and $\sigma^2$ is known. In this case, 
the test statistics are iid $N(0,1)$ variables and are written with $z$. 
Furthermore, in the orthogonal setting test statistics and p-values do not 
change depending on the model in which they are estimated. Therefore, all 
statistics are assumed to be computed in simple regressions: $M = \emptyset$. 
The z-statistics  for $H^{\emptyset,1},\ldots,H^{\emptyset,10}$ are 
$z^{\emptyset,1},\ldots,z^{\emptyset,10}$ with corresponding p-values 
$p^{\emptyset,1},\ldots,p^{\emptyset,10}$. The feature selection problem is 
equivalent to determining an order for testing $H^{\emptyset,[10]}$ while 
controlling a measure of false rejections at level $\alpha$. Since our goal is 
model selection, a feature is ``included'' or ``added'' to the model when the 
corresponding null hypothesis is rejected. Sort the hypotheses by their 
p-values $p^{\emptyset,(1)}<\ldots<p^{\emptyset,(10)}$. At step $i$, forward 
stepwise tests $H_i = H^{\emptyset,(i)}$ using test statistic $z_i$.

As expected, the distributions of the absolute order statistics are 
significantly different than the naive $|N(0,1)|$. Figure \ref{fig:effects-1} 
and \ref{fig:effects-2} show the distributions of $|z_1|$ and $|z_3|$, the 
magnitude of statistics chosen in steps 1 and 3. Informally, the difference 
between these distributions and the distribution of $|N(0,1)|$ is the ranking 
effect. This name is motivated as the difference between the test of a rank 
statistic and a randomly chosen one. Since our goal is not to estimate the 
correct distribution but to perform a valid, two-sided test, we desire a 
critical value yielding a level-$\alpha$ test. The nominal $\alpha=.1$ critical 
value is 1.645, whereas the simulated threshold is 2.58. This value can be 
easily computed using the Bonferroni correction, and the asymptotic, expected 
size of further rank statistics can be computed in the orthogonal case 
\citep{GeorgeF00}.

\begin{figure}
\centering
  \begin{subfigure}{.3\textwidth}
   \centerline{\includegraphics[width=1\textwidth]{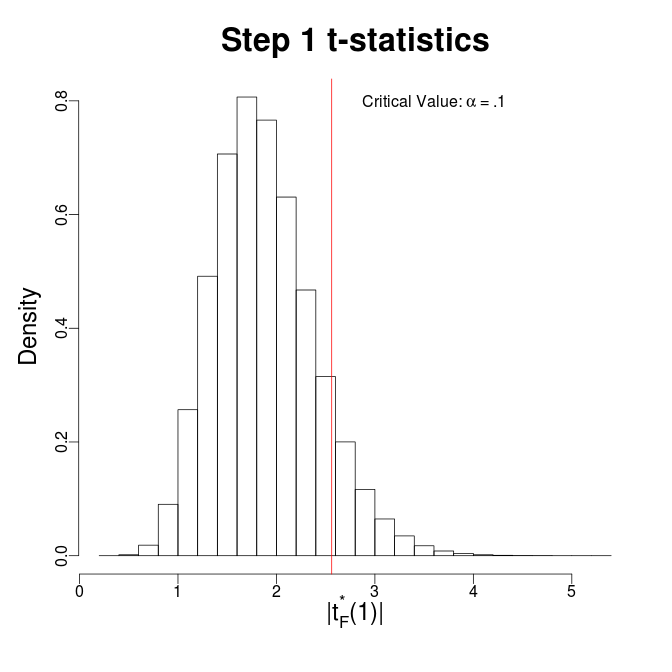}}
   \caption{}
   \label{fig:effects-1}
  \end{subfigure}
  \begin{subfigure}{.3\textwidth}
   \centerline{\includegraphics[width=1\textwidth]{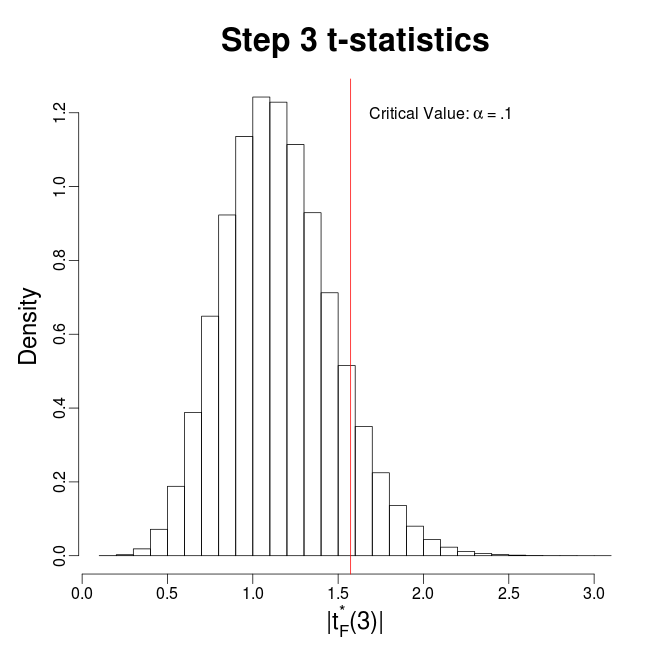}}
   \caption{}
   \label{fig:effects-2}
  \end{subfigure}
  \begin{subfigure}{.3\textwidth}
  \centerline{\includegraphics[width=1\textwidth]{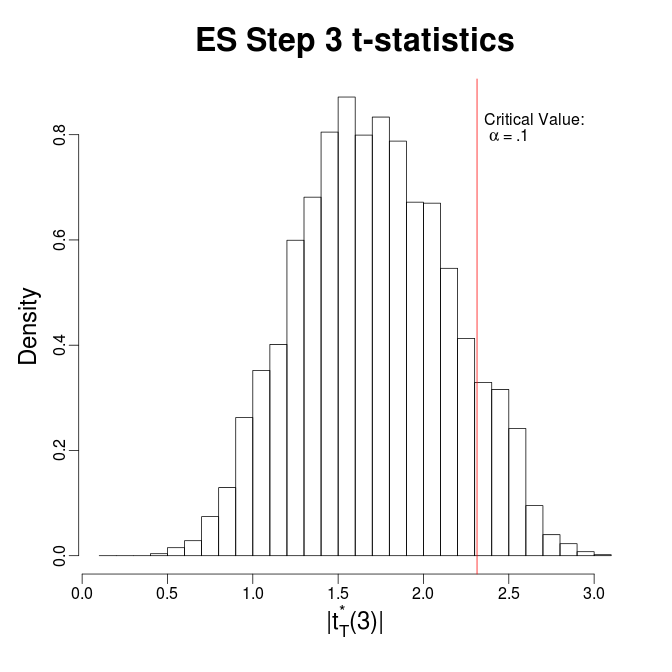}}
  \caption{}
   \label{fig:effects-3}
  \end{subfigure}
  \caption{Illustration of selection and sequential effects under the global
null hypothesis.}
  \label{fig:effects}
\end{figure}

Consider a procedure in which forward stepwise is terminated on the first step 
in which a hypothesis fails to be rejected. In this case, $H_i$ is only tested 
if $H_{[i-1]}$ are all rejected. This procedure is discussed in 
\citet{BrownJ16comm}. The simulated, .1-critical value for $|z_3|$ is 
approximately 1.57, which is lower than the naive level-.1 significance 
threshold. On one hand, this is intuitive as $|z_3|$ is constrained to 
be less than $|z_1|$ and $|z_2|$ by definition. That being said, we are only 
testing $|z_3|$ on the subset of cases in which both $H_1$ and $H_2$ are
rejected. On this subset of cases, both $|z_1|$ and $|z_2|$ are 
large, thus the constraint $|z_3| < |z_2| < |z_1|$ does not place as strong of 
a restriction on $|z_3|$.

The relevant distribution of $|z_3|$ is only realized on the subset of cases in 
which $H_3$ is actually tested. Figure \ref{fig:effects-3} shows the 
distribution of $|z_3|$ on the subset of cases in which $H_1$ and $H_2$ were 
rejected using $\alpha=.1$ using thresholds from Holm's step-down procedure 
\citep{Holm79}: given $m$ hypotheses, the p-value threshold for $p_{(i)}$ is 
$\frac{\alpha}{m-i+1}$, for $i \in \{1,\ldots,m\}$. Informally, the difference 
between the distributions in Figures \ref{fig:effects-2} and 
\ref{fig:effects-3} 
is the testing effect. The testing effect increases the simulated critical 
value 
from 1.57 to 2.32.  The remainder of this section develops a procedure which 
generates correct critical values for this setting. It is instructive to 
continue in the orthogonal setting before addressing the general case.


\subsection{Orthogonal Case}
\label{sec:orthog}


The pseudocode for a threshold approximation to stepwise using the Holm levels, 
called Revisiting Holm (RH),  is given in Algorithm  \ref{alg:RH}. As the index 
of the test is irrelevant here, the  subscript is removed. In words, RH 
``passes'' through the features multiple times, testing all features at levels 
determined by the Holm procedure. The testing pass is indexed by $s$ and will 
also be referred to as a ``round'' of testing. For now, assume that only one 
rejection is made per round, ie, $p_{(s)} \leq \frac{\alpha}{m-s+1}$ and 
$p_{(s+1)} > \frac{\alpha}{m-s+1}$. The procedure terminates when either no 
rejections are made in a single testing pass or all hypotheses have been 
rejected. Note that the selected model includes the features corresponding to 
the rejected hypotheses. 

\begin{algorithm}
\caption{Revisiting Holm (RH)}
\label{alg:RH}
\begin{algorithmic}
  \STATE {\bfseries Input:} Feature matrix $\X$, response $Y$
  \STATE {\bfseries Output:} Model corresponding to a set of features $M 
\subset 
[m]$
  \STATE {\bfseries Set:} $M = \emptyset$, $s = 1$.
  \WHILE {$|M| \leq m$}
   \FOR[Loop is a testing ``round'' or ``pass'']{$j$ in $[m] \backslash M$} 
    \IF {$H^{M,j}$ rejected at level-$\frac{\alpha}{m-s+1}$}
      \STATE $M = M \cup j$ 
    \ENDIF
   \ENDFOR
   \IF {No rejections in testing round}
    \STATE {\bfseries Return:} $M$  \COMMENT{Early termination}
   \ENDIF
   \STATE $s$ = $s$ + 1 \COMMENT{Next testing pass}
  \ENDWHILE
  \STATE {\bfseries Return:} $M = [m]$ \COMMENT{All hypotheses rejected} 
\end{algorithmic}
\end{algorithm}

With orthogonal data and independent $\hat \sigma$, the model $M$ in which  
features are tested is irrelevant as neither hypotheses nor test statistics 
change as $M$ does. It is included in the notation for generality in later 
sections. For clarity, this subsection will therefore index hypotheses as 
$H^{\emptyset,j}$ to highlight the fact that they do not change. 
\citet{FosterS08} note that this procedure produces thresholds  similar to the 
Benjamini-Hochberg (BH) procedure \citep{BenH95}. The current work provides the 
modifications 
necessary to use this concept as a valid model selection procedure in 
non-orthogonal settings and provides general performance guarantees.

While Algorithm \ref{alg:RH} may appear effectively identical to the original 
Holm procedure, there is an important distinction: RH formally tests hypotheses 
multiple times. In the classical use of the Holm procedure in which hypotheses 
are not tested mutliple times, a level $\frac{\alpha}{m-s+1}$ test for 
$H^{\emptyset,j}$ simply rejects if $p_j < \frac{\alpha}{m-s+1}$. In our case, 
hypothesis tests in later rounds must account for the failed test in previous 
rounds.

In our simulation example with $m=10$ and $\alpha=.1$, the second pass performs 
a level-$.01/9$ test conditional on the p-value being greater than the first 
pass threshold of $.1/10$. Under the null hypotheses, the sequential p-value is 
still uniformly distributed; hence the rejection threshold for the second pass 
is $p_2^* = .021$. We summarize this simple calculation for general use with a 
definition:
\begin{defn}[Conditional Rejection Threshold]
 \label{defn:nibble}
The rejection threshold, $p^*$ for a level $\alpha$-test of $H_i$ given that it 
failed to be rejected in a previous test with threshold $p_1$ is 
\begin{IEEEeqnarray}{rCl}
 p^* & = & p_1 + \alpha - p_1\alpha. \label{eqn:nibble}
\end{IEEEeqnarray}
Stated differently using $p^*$ as above,
\begin{IEEEeqnarray*}{rCl}
 \alpha & = & \Prob_{H_i}(p \leq p^*|p > \alpha_1)
\end{IEEEeqnarray*}
\end{defn}

For clarity, Figure \ref{fig:rev-holm} shows the first few testing passes of 
RH. 
Our hypothetical data follows the simple simulation example with 10 orthogonal 
explanatory variables and $\alpha=.1$. The rejection thresholds during the 
first 
four passes are the horizontal, dashed lines. The first round of the procedure 
tests all p-values at level .1/10. As one p-value falls below this threshold, 
the corresponding hypothesis is rejected and the procedure continues. Round 2 
tests the remaining hypotheses at level-.1/9 which leads to a rejection 
threshold of .021. One p-value is below this threshold, so its hypothesis is 
rejected and the procedure continues. Round 3 tests the remaining hypotheses at 
level-.1/8 and rejection threshold .033, which leads to a third rejection. 
Round 
4, however, fails to make any rejections using a rejection threshold of .047. 
Therefore the algorithm terminates, resulting in the model selected during the 
first 3 rounds: features 2, 4, and 6.
\begin{figure}
 \centerline{\includegraphics[width=.4\textwidth]{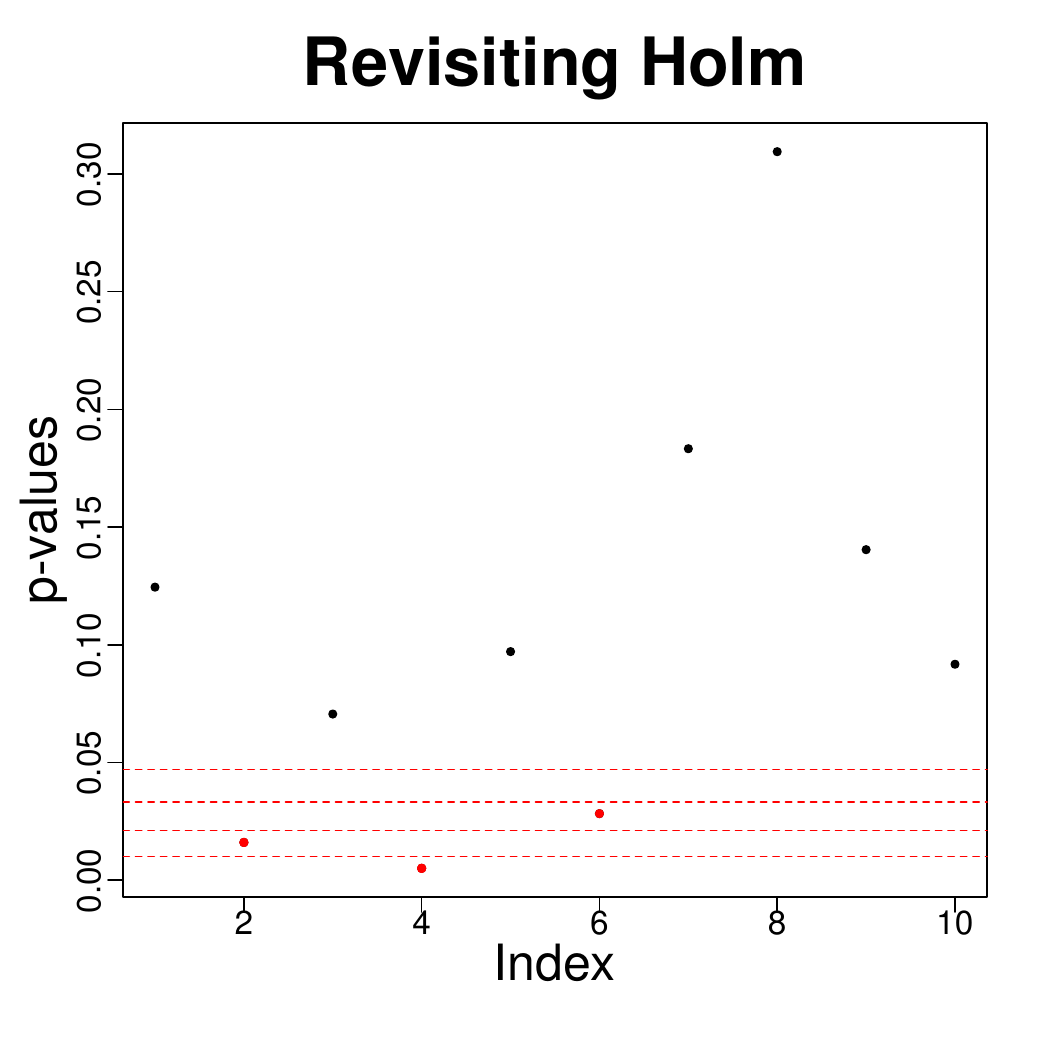}}
 \caption{Example of Revisiting-Holm procedure.}
 \label{fig:rev-holm}
\end{figure}

If only one hypothesis is rejected per round, then RH exactly replicates the 
forward stepwise selection path. To relax this assumption, suppose that 
$p_{(1)}<\alpha/m$ and $p_{(2)}<\alpha/m$, such that both hypotheses would be 
rejected on the first pass. Since $H^{\emptyset,(1)}$ was rejected, 
$H^{\emptyset,(2)}$ could have been tested at level-$\alpha/(m-1)$. Such a test 
has higher power, but was ultimately unnecessary; the conservative test 
conducted in the first round successfully rejected $H^{\emptyset,(2)}$.

If multiple hypotheses are rejected in a round, RH is not guaranteed to have 
selected the most significant feature first. Since the features were not truly 
sorted, it is unknown which of the two hypotheses rejected in the first pass 
actually had a smaller p-value; both p-values were merely smaller than 
$\alpha/m$. Selecting features in the wrong order is not of serious concern in 
the orthogonal case, because the same set of features will have been selected 
by 
the end of each testing pass. In nonorthogonal settings, however, test 
statistics change based on the model in which they are computed, so selecting 
features in a different order can lead to significantly different models.

\subsection{General Case: Nonorthogonal Data}

In nonorthogonal data, including covariates in the incorrect order can lead to 
significantly different models being selected. This is due to the test 
statistic 
being model dependent and is easiest to see by example. Table 
\ref{tab:step-order} gives the sequential p-values of all features in the 
prostate cancer data in different selected models. Two algorithms are compared: 
RH testing the features in sorted stepwise order (1), (2),$\ldots$,(8) 
(RH-sort), and RH testing the features in the reverse order (8), (7),$\ldots$, 
(1) (RH-rev). The reverse order provides a worst-case ordering for RH. In the 
table, hyphens indicate the features in the model.

\begin{table}
\centering
 \caption{Stepwise p-values after each step.}
 \label{tab:step-order}
 \begin{tabular}{r|ccc|ccc}
  \hline
 Feature (Step) & RH-0 & RH-sort-1 & RH-sort-2&RH-rev-1 &RH-rev-2 &RH-rev-3\\ 
  \hline
  lcavol (1)& 0.0000 & - & - & 0.0000 & 0.0000 & 0.0000 \\ 
  lweight (2)& 0.0000 & 0.0003 & - & 0.0000 & 0.0000 & 0.0000 \\ 
  svi  (3)& 0.0000 & 0.0410 & 0.0424 & 0.0000 & 0.0001 & 0.0000 \\ 
  lbph (4)& 0.0006 & 0.0041 & 0.1506 & 0.0010 & 0.0001 & - \\ 
  pgg45 (5)& 0.0000 & 0.1453 & 0.0758 & 0.0002 & 0.1330 & 0.1078 \\ 
  lcp (6)& 0.0000 & 0.7300 & 0.9494 & 0.0000 & - & - \\ 
  age (7)& 0.0027 & 0.7998 & 0.4649 & 0.1352 & 0.1331 & 0.7534 \\ 
  gleason (8)& 0.0000 & 0.6516 & 0.3592 & - & - & - \\ 
   \hline
\end{tabular}
\end{table}

Forward stepwise, RH-sort, and RH-rev consider the same p-values initially 
(step 
0), as no features have been added to the model. These p-values are computed 
from simple regressions between the response and the feature of interest using 
an independent estimate of the error variance. While all features fall below 
the 
RH threshold, lcavol has the lowest p-value. Therefore, forward stepwise and 
RH-sort select the same feature on step 1.
The p-values in the column RH-sort-1 are the stepwise p-values given that lcavol
is in the model. Again, RH-sort and forward stepwise select the same variable,
lweight, at the second step. Adjusting the stepwise p-values for the model
(lcavol, lweight) results in the column RH-sort-2. 
All of these p-values fall above the RH threshold for the third testing pass, so
the procedure terminates. The correspondence between RH-sort and forward
stepwise seen here is a general property: if RH tests variables in the order
determined by stepwise, then RH selects variables in the same order as stepwise.

RH-rev behaves significantly differently than RH-sort and forward stepwise. The
initial p-values it considers are identical, but RH-rev tests gleason first
and the test is rejected. The p-values in the column RH-rev-1 condition on
gleason being in the model. Proceeding in the reverse order, the test of age is
not rejected, but the test of lcp is. Column RH-rev-2 updates the stepwise
p-values given the model contains gleason and lcp. Using these p-values, lbph is
also rejected, and the process continues. In fact, RH-rev rejects all 8 
features. Given the ordering of the features, this is at least justifiable: 
each 
subsequent feature explains a significant reduction in ESS. Even after several 
features are in the model, lcavol provides unique information about the 
response. That being said, selecting all 8 features is clearly not desirable. 
The next section uses a different set of rejection thresholds and mimics 
stepwise regression precisely on these data: we identify the RH-sort model of 
\{lcavol, lweight\} regardless of the order in which features are tested.

In nonorthogonal data, one may also object to the updating done via Definition 
\ref{defn:nibble} because sequential p-values are relevantly different between 
steps. While the same explanatory feature is being tested, the null hypothesis 
and sequential p-value are model dependent. Furthermore, there is, in general, 
no guarantee that the conditioning statement in equation (\ref{eqn:nibble}) is 
accurate; a feature can become \emph{more} significant in the presence of other 
features.  The next section renders this critique obsolete by introducing a new 
series of testing thresholds: the adjustment made in equation 
(\ref{eqn:nibble}) 
makes a negligible change to the effecting testing level and can be ignored 
with a minimal reduction in power.

\section{Better Threshold Approximation}
\label{sec:rai}

To better approximate forward stepwise, initial testing passes need to search 
for more significant features. As forward stepwise searches for the feature 
which yields the maximal improvement in \Rs, we consider a procedure which 
tests 
for an increase in \Rs of $r^s$, $r \in (0,1)$ and $s$ is the testing pass. 
For example, if $r=1/2$, then the first testing pass tests for features which 
increase \Rs by 1/2, while the second pass tests for those yielding an increase 
of 1/4. This yields a geometrically decreasing sequence of bounds. By 
choosing $r$, this provides a set of algorithms which are collectively referred 
to as Revisiting Alpha-Investing (RAI).  In order to specify the stopping 
criterion and fully describe RAI, we need to introduce alpha-investing 
\citep{FosterS08}. Afterward, we provide our performance guarantee.

\subsection{Revisiting Alpha-Investing}
\label{app:mfdr-control}


Alpha-investing rules are similar to alpha-spending rules in that they are 
given 
an initial amount of alpha-wealth to be spent on hypothesis tests. Wealth can 
be 
considered as an allotment of error probability. Bonferroni allocates this 
error 
probability equally over all hypothesis, testing each one at level $\alpha/m$. 
In general, the amount spent on tests can vary. If $\alpha_i$ is the amount of 
wealth spent on test $H_i$, FWER is controlled when \[\sum_{i=1}^m  \alpha_i 
\leq \alpha.\] 

In clinical trials, alpha-spending is useful due to the varying
importance of hypotheses. For example, many studies include both primary and
secondary endpoints. As the primary endpoint is the most important
hypothesis, the majority of the alpha-wealth can be spent on it, providing
higher power. Alpha-spending rules can allocate the remaining wealth equally 
over the secondary hypotheses. FWER is controlled and the varying importance of 
hypotheses is acknowledged. Interested readers are referred to 
\citet{DmitrienkoTB10} and references therein for more on FWER control 
procedures.

Alpha-investing rules are similar to alpha-spending rules except that
alpha-investing rules earn a return, or contribution to their alpha-wealth, of
$\omega\leq\alpha$ when tests are rejected. Therefore, the alpha-wealth after
testing hypothesis $H_i$ is 
\begin{IEEEeqnarray*}{rCl}
 W_{i+1} & = & W_i - \alpha_i + \omega R_i
\end{IEEEeqnarray*}
An alpha-investing strategy uses the current wealth and the history of previous
rejections to determine which hypothesis to test and the amount of wealth that
should be spent on it.

Intuitively, alpha-investing rules spend error probability in search of false
null hypotheses to reject. Each false null that is rejected allows $\alpha$ more
incorrect rejections in expectation. Alpha-investing rules merely need to spend
more wealth (error probability) than the probability of error they incur. In
some sense, this behavior is present in all procedures which control a
proportion of false rejections. For example, if it is known that the first 9
rejections were of false hypotheses, then any 10th hypothesis can be rejected
while controlling the proportion of false rejections at .1.

The only assumption an alpha-investing rule requires in order to control mFDR 
is 
that each test must be conditionally level-$\alpha$ given the sequence of 
rejections: \[\E(V_i|R_1,\ldots,R_{i-1}) \leq \alpha_i.\] \citet{FosterS08} 
assume tests are independent. Theorem \ref{thm:valid-tests} allows 
alpha-investing to be used in many more scenarios. Therefore, once Theorem 
\ref{thm:valid-tests} is proven, RH and RAI control mFDR by virtue of being 
alpha-investing rules.

Viewing the Holm step-down procedure as an alpha-investing rule yields RH. Given 
initial alpha-wealth $\alpha$ and return $\omega=\alpha$, test all hypotheses at 
the Bonferroni level, $\alpha/m$. This exhausts all alpha-wealth, so that the 
procedure terminates if no rejections are made. If a rejection is made, the 
procedure earns a return equal to $\alpha$ and only $m-1$ hypotheses remain. The 
wealth is again split evenly among all remaining hypotheses, yielding the 
Bonferroni threshold over $m-1$ hypotheses of $\alpha/(m-1)$. If no rejections 
are made in a round, then the alpha-investing rule is out of wealth and the 
algorithm terminates. 

RAI merely provides a different sequence of levels at which to test. Psuedocode 
for the procedures is given in Algorithm \ref{alg:RAI}. Note that the 
implementation in R \citep{R_2019, rai_R} makes slight modifications to this for 
practical performance improvements. Most importantly, it uses a conservative 
estimate of error variance $\hat \sigma$ by using the residuals from model $M$ 
instead of $M \cup j$. This prevents $\hat\sigma$ from being recomputed for 
every hypothesis test. While there is a concomitant 
loss of power, the performance improvement is significant and we still observe 
strong results. RAI is well defined in any model in which it is possible to test 
the addition of a single feature such as generalized linear models. The testing 
thresholds ensure that the algorithm closely mimics forward stepwise, which 
provides the performance guarantees of the next subsection.
\begin{algorithm}
 \caption{Revisiting Alpha-Investing (RAI)}
 \label{alg:RAI}
 \begin{algorithmic}
  \STATE {\bfseries Input:} Feature matrix $\X$, response $Y$, threshold level 
$r$, mFDR level $\alpha$ 
  \STATE {\bfseries Output:} Model corresponding to a set of features $M 
\subset 
[m]$
  \STATE {\bfseries Set:} $M = \emptyset$, $s = 1$, wealth = $\alpha$
  \WHILE {$|M| \leq m$}
  \STATE Set $\alpha_s = f(r^s)$ \COMMENT{Alpha determined from \Rs threshold; 
See Appendix} 
  \FOR[Loop is a testing ``round'' or ``pass'']{$j$ in $[m] 
\backslash M$}
  \STATE wealth = wealth - $\alpha_s$ \COMMENT{Pay for test} 
  \IF {wealth $<$ 0}
  \STATE {\bfseries Return:} $M$  \COMMENT{Early termination}
  \ENDIF 
  \IF {$X_j$ increases $\text{R}^2_M$ by $r^s$}
  \STATE $M = M \cup j$; wealth = wealth + $\alpha$
  \ENDIF
  \ENDFOR
  \STATE $s$ = $s$ + 1 \COMMENT{Next testing pass}
  \ENDWHILE
  \STATE {\bfseries Return:} $M = [m]$ \COMMENT{All hypotheses rejected} 
 \end{algorithmic}
\end{algorithm}

Approximating stepwise using these thresholds has many practical performance
benefits. First, multiple passes can be made without any rejections before the
algorithm exhausts its alpha-wealth and terminates. The initial tests are
extremely conservative but only spend tiny amounts of wealth; however, tests
rejected in these stages still earn the full return $\omega$. This ensures that
wealth is not wasted too quickly when testing true null hypotheses. Furthermore,
false hypotheses are not rejected using significantly more wealth than is
required. An alternative construction of alpha-investing makes this latter
benefit explicit and is explained in \citet{FosterS08}. Taken together, this
improves power in ways not addressed by the theorem in the next section. By
earning more alpha-wealth, future tests can be conducted at higher power while
maintaining mFDR control.

RAI performs a sequential search for sufficient model improvement as opposed to
the global search for maximal improvement performed by forward stepwise. Most
sequential, or online, algorithms are online in the observations, whereas RAI is
online in the \emph{features}. This allows features to be generated dynamically
and allows extremely large data sets to be loaded into RAM one feature at a
time. As such, RAI is trivially parallelizable in the MapReduce setting, similar
to \citep{Kumar+13}. For example, many processors can be used, each considering
a disjoint set of features. Control need only be passed to the master node
when a significant feature is identified or a testing pass is completed.
Parallelizing RAI will be particularly effective in extremely sparse models,
such as those considered in genome-wide association studies. Online feature
generation is beneficial when features are costly to generate and can be used
for directed exploration of complex spaces. This is particularly useful when
querying data base or searching interaction spaces as in Section
\ref{sec:poly}.

Variance inflation factor regression (VIF) \citep{LinFU11} computes stepwise 
t-statistics extremely quickly with little loss in accuracy. With this 
enhancement, RAI performs forward stepwise and model selection in $O(nm\log(n))$ 
time as opposed to the $O(nm^2q^2)$ required for traditional forward stepwise, 
where $q$ is the size of the selected model. The log term is an upper bound on 
the number of testing passes performed by RAI. This is significantly reduced for 
large $n$ by recognizing when passes may be skipped, which is possible whenever 
a full pass is made without any rejections. The control provided by 
alpha-investing is maintained, because RAI must pay for all of the skipped 
tests. Using this computational shortcut, we find that only 7-10 passes are 
required to select a model using RAI.

 
\subsection{Approximation Guarantee}

This subsection bounds the performance of RAI and requires additional notation. 
We will often need to consider a feature $\X_i$ orthogonal to those currently 
in 
the model, $\X_M$. This will be referred to as adjusting $\X_i$ for $\X_M$ and 
the corresponding feature is denoted $\X_{i.M} = P_{M^\perp}\X_i$. 
This same notation holds for sets of variables: $\X_A$ adjusted for $\X_M$ is 
$\X_{A.M} = P_{M^\perp}\X_A$.

RAI is proven to perform well if the improvement in fit obtained by adding a 
set 
of features to a model is upper bounded by the sum of the improvements of 
adding 
the features individually. If a large set of features improves the model fit 
when considered together, this constraint requires some subsets of those 
features to improve the fit as well. Consider the improvement in model fit by 
adding $\X_S$ to the model $\X_M$: \[\Delta_M(S) := \mRs(S \cup M) - \mRs(M).\] 
Letting $S = A \cup B$, we bound $\Delta_M(S)$ as
\begin{IEEEeqnarray}{rCl}
 \Delta_M(A) + \Delta_M(B) & \geq & \Delta_M(S). \label{eqn:submod.int} 
\end{IEEEeqnarray}

If $A\cup B$ improves the model fit, equation (\ref{eqn:submod.int})
requires that either $A$ or $B$ improve the fit. Therefore, signal that is
present due to complex relationships among features cannot be completely hidden
when considering subsets of these features.
Equation (\ref{eqn:submod.int}) defines a submodular function:
\begin{defn}[Submodular Function]
 Let $F: 2^{[m]} \rightarrow \bbR$ be a set function defined on the the power
 set of
 $[m]$. $F$ is submodular if $\forall A,B \subset [m]$
 \begin{IEEEeqnarray}{rCl}
  F(A) + F(B) & \geq & F(A \cup B) +F(A \cap B) \label{eqn:submod}
 \end{IEEEeqnarray}
\end{defn}
This can be rewritten in the style of (\ref{eqn:submod.int}) as
\begin{IEEEeqnarray*}{rCl}
 F(A) - F(A\cap B) + F(B) - F(A \cap B) & \geq & F(A \cup B) - F(A \cap B)\\
 \Rightarrow \Delta_{A\cap B}(A) + \Delta_{A\cap B}(B) 
 & \geq & \Delta_{A\cap B}(A\cup B),
\end{IEEEeqnarray*}
which considers the impact of $A \backslash B$ and $B \backslash A$ given $A
\cap B$. Given (\ref{eqn:submod.int}), it is natural to approximate the
maximizer of a submodular function with a greedy algorithm. We provide a proof
of the performance of RAI by assuming that \Rs is approximately submodular. 

In order for these results to hold even more generally, the definition of 
submodularity can be relaxed \citep{DasK11}. To do so, iterate 
(\ref{eqn:submod.int}) until the left hand side is a function of the influences 
of individual features and only require the inequality to hold up to a 
multiplicative constant $\gamma \geq 0$. Given a model $M$, consider 
adding the features in $A=\{a_i,\ldots,a_l\} \subset [m]\backslash M$. 
Hence $\Delta_M(a_i)$ is the marginal increase in \Rs by adding $a_i$ to model 
$M$. When data is normalized, $\Delta_M(a_i)$ is the squared 
partial-correlation 
between the response $Y$ and $a_i$ given $M$: $\Delta_M(a_i) = \text{Cor} 
(Y,a_{i.M})^2$. Define the vector of partial correlations as 
$r_{Y,A.M} = \text{Cor}(Y, A.M)$, then the sum of individual contributions to 
\Rs is $\|r_{Y,A.M}\|_2^2$. Similarly, if we define $C_{A.M}$ as the 
correlation 
matrix of $A.M$, then $\Delta_M(A) = r_{Y,A.M}'C_{A.M}^{-1}r_{Y,A.M}$.

\begin{defn}(Submodularity Ratio)
 The submodularity ratio, $\gamma_{sr}$, of \Rs with respect to a model $M$ and 
$k \geq 1$ is
 \begin{IEEEeqnarray*}{rCl}
  \gamma_{sr}(M,k) & = & \min_{(S:S\cap M = \emptyset, |S| \leq k)}
  \frac{r_{Y,S.M}'r_{Y,S.M}}{r_{Y,S.M}'C_{S.M}^{-1}r_{Y,S.M}}
 \end{IEEEeqnarray*}
\end{defn}
The minimization identifies the worst case set $S$ to add to the model $M$. It 
captures how much \Rs can increase by adding $S$ to $M$ (denominator) compared 
to the combined benefits of adding its elements to $M$ individually 
(numerator). 
If $M$ is the size-$k$ set selected by forward stepwise, then \Rs is 
approximately submodular if $\gamma_{sr}(M,k) > \gamma$, for some constant 
$\gamma > 0$. We will refer to data as being approximately submodular if \Rs is 
approximately submodular on the data. \Rs is submodular if $\gamma_{(M,2)} \geq 
1$ for all $M \subset [m]$ \citep{John+15sub}. This definition is extremely 
similar to that of \citet{DasK11}, but slightly more refined.

Our main theoretical result provides a performance guarantee for a slightly 
modified version of RAI.
RAI has to be modified in order to account for the fact that 
there is, in general, no constraint on the behavior of p-values for tests of 
the 
same feature in different models. The submodularity ratio will allow us to make 
some claims about the improvement in \Rs when adding sets of features, but the 
control it provides on the change of individual p-values is quite poor, 
particularly if the model has changed significantly between two tests of the 
same feature. While the previous section demonstrates that we can still use the 
p-value in order to test a hypothesis, it does not specify any relationship 
between the observed p-values for testing $H^{M,j}$ and $H^{M',j}$, when $M 
\neq M'$.

The modified procedure, \RAIp, is almost identical to RAI, in that it uses an 
increasing sequence of threshold values and tests features sequentially. The 
testing threshold, however, is not increased until all features fail to be 
rejected in a \emph{single} model. The performance of this procedure is 
effectively identical to RAI. That being said, the performance guarantee 
requires a single model in which features can be compared. Therefore, we will 
refer to a testing pass of \RAIp as all tests using a given threshold. This 
single pass may cycle through all of the features multiple times. In the worst 
case, \RAIp cycles through all features in order to reject only one feature, 
thus not making any computational improvement over stepwise regression. This, 
however, is highly unlikely and is never observed in our examples. Furthermore, 
once all features have been tested in the same model, \RAIp can skip to the 
round in which the next feature would be rejected. Therefore in practice, \RAIp 
and RAI perform approximately the same number of computations.

We restate the essential components of Theorem \ref{thm:performance} to make 
the 
subsequent discussion easier to follow. The model selected by \RAIp, $M_l$, 
satisfies $\mRs(M_l) \geq (1-e^{l/c})\mRs(M_k^*)$, where $c = \left( 
\frac{\iota+k}{\gamma r} - \iota \right)$ and $\iota$ is the maximum number of 
features rejected in a testing pass. While the proof is deferred to Appendix 
\ref{app:performance}, a few remarks are in order.
\begin{enumerate}
 \item This bound holds for any number of rejected features $l$. Therefore, the 
result in some sense mirrors the results of alpha-investing in which type-I 
error is controlled at any stopping time. It is more flexible, however, as $l > 
k$ can be considered. In this way, one can use RAI or \RAIp to over-estimate 
the 
support of the true model while still maintaining a performance guarantee.
 \item In usual application, $l$ is actually chosen adaptively and $M_l$ would 
presumably include many of the features of $M_k^*$. Neither of these facts are 
leveraged in the proof; it is a worst-case guarantee assuming that $M_l \cap 
M_k^* = \emptyset$.
 \item The value $\iota$ is upper bounded by $l$, but in practice is far 
smaller. Furthermore, it is computed while running \RAIp. Therefore the exact 
value can be used in the bound after the procedure has terminated. 
Alternatively, $l$ could be chosen such that $\iota$ is small. This may be 
helpful when the last testing pass rejects a large number of features, for 
example.
 \item The proof demonstrates an important fact that also arises in SURE 
\citep{FanL08}: it is important to bound the effect of adding sets of features 
at the same time, not the effect of adding individual features. Stronger claims 
could be made if we were willing to make the style of assumptions in SURE, in 
which they bound characteristics of each individual feature. Instead, we opt 
for 
a bound on the submodularity ratio, which is weaker but still provides a 
performance guarantee. \end{enumerate} 

\subsection{Exact Forward Stepwise}
\label{sec:exactFS}

In this section, we further examine the setting $r \to 1$. This allows RAI to 
exactly mimic 
forward stepwise. Furthermore, the algorithms need not actually be run as the 
resulting behavior can be computed in closed form.  As $r \to 1$, RAI conducts 
tests at level $\alpha_s \to 0$, $\forall s$. For concreteness, suppose that 
$\alpha_s = \delta >0$, $\forall s$. Note that repeatedly testing a hypothesis 
at level $\delta$ leads to an approximately linear increase in the rejection 
threshold by Definition \ref{defn:nibble}, because $2\delta - \delta^2 \approx 
2\delta$. For sufficiently small $\delta$, this procedure selects variables in 
the same order as forward stepwise. 

Consider the amount of wealth spent to reject a null hypothesis $H_0$ if its 
p-value is $p_0$. Each failed rejection implies that $p_0$ is in the upper 
$(1-\delta)$ portion of its feasible region, which is initially $[0,1]$. If 
$H_0$ is rejected after $q$ tests, a Taylor approximation provides
\begin{IEEEeqnarray*}{rCl}
 (1-\delta)^q & = & 1-p_0\\
 \Rightarrow q\delta & \approx & -\log(1-p_0). 
\IEEEyesnumber\label{eqn:eps.nibble}
\end{IEEEeqnarray*}
While $H_0$ could have been rejected by spending $p_0$ initially, $-\log(1-p_0) 
> p_0$ was spent on the rejection. If $p_0$ is small, the amount of 
alpha-wealth 
wasted by revisiting is minor, but larger p-values waste significant wealth. 

Combining the above claims, the results of RAI can be derived in closed form. 
Since hypotheses are rejected in the stepwise order using the sequential 
p-values, suppose $H_{[m]}$ is the set of hypotheses ordered by forward 
stepwise 
with corresponding p-values $p_{[m]}$. The first hypothesis is only rejected if 
RAI spends $\delta$ to test all hypotheses until a total of $-\log(1-p_1)$ has 
been spent on each test. Similarly, the second hypothesis is only rejected if 
RAI continues to spend $\delta$ to test all remaining hypotheses until 
$-\log(1-p_2)$ has been spent. Note that we ignore the update of equation 
(\ref{eqn:nibble}) as the hypotheses are not assumed to be independent and the 
update is significant in this case. The resulting procedure, Stepwise-RAI 
(S-RAI), is given in Algorithm \ref{alg:S-RAI}.
 \begin{algorithm}
  \caption{Stepwise Revisiting Alpha-Investing (S-RAI)}
  \label{alg:S-RAI}
  \begin{algorithmic}
   \STATE {\bfseries Input:} Hypotheses determined by forward stepwise, 
$H_{[m]}$ and corresponding sequential p-values, $p_{[m]}$.
   \STATE {\bf Set:} $\hat k = \max\{ k $ s.t. $-\sum_{i=1}^l (m-i+1)\log(1-p_i) 
 < l\alpha$, $\forall l \leq k\}$
   \STATE {Reject $H_1,\ldots,H_{\hat k}$}
  \end{algorithmic}
 \end{algorithm}

Given a full set of p-values as those in Table \ref{tab:prostate}, selecting a
model using hypothesis testing requires rejecting an initial contiguous set of
hypotheses. If hypotheses are ordered numerically, $H_2$ and $H_4$ cannot be the
only rejections. The sets $\{H_1,H_2\}$ or $\{H_1,\ldots,H_4\}$ are possible
rejection sets that identify forward stepwise models. As p-values are
not necessarily sorted by size as required by the BH procedure,
controlling FDR under this constraint is nontrivial. \citet{GSell+15} transform
p-values such that they are ordered and use BH on the transformed
p-values. The model selection criteria they consider is ForwardStop (FS), which
rejects hypotheses
$H_1,\ldots,H_{\hat k}$ where
\begin{IEEEeqnarray}{rCl}
 \hat k & = & \max_{k\in\{1,\ldots,m\}} -\frac{1}{k}\sum_{i=1}^k \log(1-p_i) 
\leq \alpha. \label{eqn:fs}
\end{IEEEeqnarray}

Observe that the p-value transformations in equations (\ref{eqn:eps.nibble}) 
and 
(\ref{eqn:fs}) are identical. The conversion necessary to apply FS is 
equivalent 
to spending wealth in a wasteful way. This wastefulness is one explanation for 
the extremely conservative behavior of FS. Given the adjusted p-values in Table 
\ref{tab:prostate} are also conservative in large regions of the parameter  
space, it is not surprising that the combined procedure is highly conservative 
as demonstrated in the next section. 

Under independence, S-RAI simplifies and looks very similar to FS, because the 
wealth spent on the unrejected hypotheses can be carried over to further rounds. 
Therefore, the second hypothesis would be rejected by S-RAI when an additional 
$-\log(1-p_2)-\log(1-p_1)$ has been spent. The result of this simplification is 
given in Algorithm \ref{alg:SI-RAI}. This provides a second way of viewing FS 
through a sequential lens: it performs S-RAI assuming independence but adapts to 
the correct subset of features. Specifically, if there were only $\hat k$ 
features instead of $m$ features, the stopping rule for FS and Algorithm 
\ref{alg:SI-RAI} would be identical. 

\begin{algorithm}
 \caption{Stepwise Revisiting Alpha-Investing (S-RAI): Independent P-values}
 \label{alg:SI-RAI}
 \begin{algorithmic}
  \STATE {\bfseries Input:} Hypotheses determined by forward stepwise, 
$H_{[m]}$ 
and corresponding sequential p-values, $p_{[m]}$.
  \STATE {\bf Set:} $\hat k$ = $\max\{ k | -\sum_{i=0}^{l-1} \log(1-p_i) + 
(m-k+1)-\log(1-p_l) < l\alpha, \forall l \leq k\}$
  \STATE {Reject $H_1,\ldots,H_{\hat k}$}
 \end{algorithmic}
\end{algorithm}
 
\subsection{Comparing Methods}
\label{sec:compare}

\mynotesH{need to check for other methods to potentially include}

Before comparing the performance of the algorithms on data, we must discuss the 
interpretation of sequential tests. When features are correlated, the truth 
value of the null hypothesis for testing feature $j$ in model $M$, $H^{M,j}$, 
may depend on $M$. Hence $H^{M,j}$ may be false given the currently model M, 
and $\X_j$ may be correctly included at that step. Within a later active model, 
$M' \supset M$, $H^{M',j}$ may be true. Thus, a correct selection at a given 
step may become ``incorrect'' as the process proceeds, and vice-versa. This 
phenomenon is due to deviations from submodularity and is often called 
suppression \citep{John+15sub}. Our measures of false rejections are 
\emph{model dependent}. We are not penalized if a feature which was 
correctly rejected is later deemed an incorrect rejection: the quality of a
decision is determined at the time in which the decision was made.


While there are concerns over the ``exact'' description of the adjusted 
p-values 
in Table \ref{tab:prostate}, one could still use them for model selection. This 
is particularly salient as the R package selectiveInference
\citep{selInf_R} which implements the procedure reports a statistic that does 
precisely that. Similarly, RH is only exact under orthogonality. We are 
nevertheless interested in their empirical performance when that assumption is 
violated. Conversely, RAI and \RAIp use intentionally conservative tests by not 
adjusting for selection as in Definition \ref{defn:nibble}. 

Our simulated data has $n=400$ observations, $m \in \{100, 200\}$ features, of 
which 
$k \in \{20, 40\}$ have a non-zero parameter value. Each feature with a 
non-zero 
parameter has a correlated counterpart with a parameter of zero. We consider 
correlations between each pair of features $\rho \in \{.2,.8\}$. The remaining 
features are orthogonal to all others in the population model. As $\X$ is 
generated randomly, the maximal correlation between such orthogonal features in 
a given data set is still approximately .2. 


Signal strengths were chosen to be close to the RIC threshold \citep{FosterG94}. 
 The high-signal case sets the magnitude of non-zero parameters equal to 
$\sqrt{4\log(m)}/\sqrt{n}$ and the low-signal case sets them to 
$\sqrt{2\log(m)}/\sqrt{n}$. True features in the high-signal case have 
t-statistics in the true model in the range $[2,7]$, while the low-signal case 
produces t-statistics in the true model in the range in the range $[1,5]$. 
As a 
final complication, the signs of the nonzero parameters alternate as this helps 
produce deviations from submodularity.

Comparisons are made between forward stepwise models selected via six 
procedures: forward stepwise stopped by CP, RH, RAI, \RAIp, FS, and Stepwise 
Holm (equivalently Max-t or SH). Stepwise Holm merely compares the sequential 
p-value to the Holm threshold without any adjustments for selection. Since RH, 
RAI, and \RAIp depend on the order in which features are tested, we provided a 
worst case ordering: features are reordered according to the stepwise selection 
path, such that $\X_i$ is selected on the $i$'th step. Features are then tested 
from $\X_m$ to $\X_1$. This ensures that all ``incorrect'' features must be 
tested before the ``correct'' stepwise feature. 

We report four statistics as performance measures computed over 100 repetitions 
of each data generating scenario: FDR, mFDR, power, and the proportion of 
stepwise features selected. Note that an individual false rejection is a 
function of our null hypothesis:  if the correct feature in a group has not 
been included, then including the correlated counterpart is not considered a 
false rejection. Due to the correlation between features, such a feature 
legitimately improves the predictive performance of the resulting linear model. 
This is the same definition of false rejections used by \citet{GSellHT13}.

We compute power as the proportion of features with a non-zero parameter 
value in the true model that are included in the final model. This allows us 
to consider correct selections in a classical setting in which a true model 
$M^*$ with parameters $\beta^*_{M^*}\ne0$ exist. If an algorithm selects model 
$M$, then 
\[\text{Power} := \frac{|M \cap M^*|}{|M^*|}.\]

The proportion of stepwise features selected is only reported for RH, RAI, and 
\RAIp. It measures the proportion of features in the selected model $M$ that 
are in the forward stepwise model \emph{of the same size}. Given the ordering 
of 
our data matrix $\X$, and a selected model $M$ of size $l$, this corresponds to 
\[\text{Proportion of Stepwise Features} := \frac{|\X_{1:l} \cap M|}{l}.\]

Figure \ref{fig:perf_sim} shows the FDR and power of the six methods over 
the simulation settings. The x-axis is ordered such that the power of RAI is 
decreasing left-to-right. As such, we call this dimension the ``difficulty'' of 
the data scenario and the mapping between difficulty and simulation settings is 
given in Table \ref{tab:difficulty}. FDR is shown here for its familiarity and 
to demonstrate that FDR control and mFDR control are extremely similar in 
practice. See Figure \ref{fig:perf_sim} for graphs involving mFDR.

\begin{table}[ht]
 \centering
 \caption{Simulation settings and corresponding difficulty rating determined by 
the power of RAI. Note that difficulty essentially sorts via signal strength 
then correlation.}
 \label{tab:difficulty}
 \begin{tabular}{llllr|llllr}
  \hline
k & m & Signal & $\rho$ & Difficulty & k & m & Signal & $\rho$ & Difficulty \\ 
  \hline
  40 & 100 & high & 0.2 &   1 & 40 & 100 & low & 0.2 &   9 \\ 
  40 & 200 & high & 0.2 &   2  &  20 & 100 & low & 0.2 &  10 \\ 
  20 & 100 & high & 0.2 &   3  & 40 & 200 & low & 0.2 &  11 \\ 
  20 & 200 & high & 0.2 &   4  & 40 & 100 & low & 0.8 &  12 \\ 
  40 & 100 & high & 0.8 &   5  & 20 & 200 & low & 0.2 &  13 \\ 
  20 & 100 & high & 0.8 &   6  & 20 & 100 & low & 0.8 &  14 \\ 
  20 & 200 & high & 0.8 &   7  & 20 & 200 & low & 0.8 &  15 \\ 
  40 & 200 & high & 0.8 &   8  & 40 & 200 & low & 0.8 &  16 \\ 
  \hline
 \end{tabular}
\end{table}

\begin{figure}
 \centering
 \caption{Performance measures. The x-axis is ordered such that the power of 
RAI is decreasing left-to-right. The simulation settings corresponding to this 
ordering are in Table \ref{tab:difficulty}. To demonstrate its impact, RAI uses 
the conditioning update of equation (\ref{eqn:nibble}) while \RAIp is 
conservative and does not.}
 \label{fig:perf_sim}
   \centerline{\includegraphics[width=\textwidth]{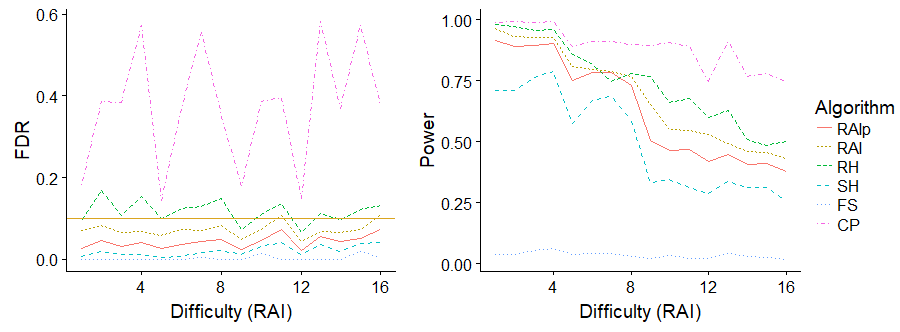}}
\end{figure}

There are many messages conveyed by Figure \ref{fig:perf_sim}. First, while
SH is conservative, it still performs very well. This is especially telling as
\citet{Taylor+14} demonstrate that the test statistic is highly conservative
after multiple steps. As pointed out in Section \ref{sec:inf4ms}, this is 
because future tests only occur if previous tests have been rejected. On this 
subset of cases, the test-statistics considered by SH are much less 
constrained. 
SH has significantly higher power than FS, which is so conservative as to 
barely reject any hypotheses in any scenario. 

Second, RAI and \RAIp have high power while controlling FDR. Given that all 
methods either select a model on the stepwise path or an approximation of it, 
there is a necessary trade-off between power and FDR. If there is a mix of 
incorrect and correct features along the stepwise path, then the only way to 
include more correct features is to make false rejections. Therefore, one must 
separate the blame, so to speak, of the performance of the procedures. Part of 
it is the result of the testing procedure, but part is due to stepwise. The 
latter is suffered by all algorithms. That being said, our revisiting 
procedures 
solve the former far better than the competitor FS. Lastly, note that the FDR 
of 
CP fluctuates so rapidly because it is sensitive to the number of features $m$.


Figure \ref{fig:perf_sim} demonstrates the control provided in Corollary 
\ref{cor:mfdr-control} and our close approximation of forward stepwise. Both 
Figure \ref{fig:perf_sim} and Figure \ref{fig:perf_sim} show that RH 
achieves high power but at the cost of losing mFDR control.  As expected, it 
does not mimic forward stepwise well; only approximately 65-70\% of the 
selected 
features are along the stepwise path of the same length. RAI and \RAIp, on the 
other hand, include a high proportion of the forward stepwise features but 
still 
control false discoveries. Note that all three revisiting procedures produce 
models that have effectively the exact same \Rs as the stepwise model of the 
same size, merely with different covariates.

\section{Searching Interaction Spaces}
\label{sec:poly}

As an application of RAI, we demonstrate a principled method to search 
interaction spaces while controlling type-I errors. This has become a question 
of increasing interest in searching for interaction effects between 
genes \citep{genesInteract11, genesInteract15}. In this case, submodularity 
is merely a formalism of the principle of marginality \citep{Nelder77}: if an 
interaction between two features is included in the multiple regression, the 
constituent features should be as well. This reflects a belief that an 
interaction is only informative if the marginal terms are as well. RAI can 
perform a greedy search for main effects, while maintaining the flexibility to 
add polynomials to the model that were not in the original feature space. 
Therefore, we search interaction spaces in the following way: run RAI on the 
marginal data $\X$; for $i,j\in[m]$, if $\X_i$ and $\X_j$ are rejected, test 
their interaction by including it in the stepwise routine. We allow $i=j$ so 
that polynomials of a single feature are also included. This bypasses the need 
to explicitly enumerate the interaction space, which is computationally 
infeasible for large problems. Furthermore, as Table \ref{tab:concrete} shows 
for the concrete compressive strength data, it can be highly beneficial to only 
consider relevant portions of interaction spaces, as the full space is often 
too 
complex. This procedure for searching interactions was also considered in 
\citet{HaoZ17}. We provide an efficient algorithm, an associated R package,
and the following performance guarantee for the method.

As \RAIp is a conservative procedure, it often makes a complete pass through 
the data using a fixed model before termination. In this case, it is easy to 
bound the improvement that any single feature can provide, and this is 
actually computed automatically by the package. While we cannot 
prove that the selected model performs as well as the optimal 
model including all possible interactions, we are able to compare to any 
possible model among the set of interactions we have currently tested. Lemma 
\ref{lem:Rsbnd} from Appendix \ref{app:performance} yields the following 
corollary:
\begin{corl}
 \RAIp selects a set of features $M$ such that adding any features $S$ 
that have been tested for addition to $M$ yields an improvement which is upper 
bounded as
  \begin{IEEEeqnarray*}{rCl}
  \mRs(M \cup S) - \mRs(M) & \leq &
  (1-\mRs(M))\frac{|S|r^{s-1}}{\gamma(M,|S|)}.
  \end{IEEEeqnarray*}
\end{corl}

\subsection{Simulated Data}
\label{sec:simulated-data}

Simulated data is used to demonstrate the ability of RAI to identify polynomials
in complex spaces. Our simulated explanatory features have the following
distribution:
\[\X_{i,j} \sim N(\tau_j,1) \quad \text{where} \quad \tau_j \sim N(0,4).\]
The true mean of $Y$, $\mu_Y$, includes four terms which are polynomials  in
the first ten marginal features:
\begin{IEEEeqnarray*}{rCl}
 Y & = & \epsilon\\
 \mu_Y & = & \beta_1\X_1\X_2 + \beta_2\X_3\X_4^2 + \beta_3\X_5\X_6^3 +
 \beta_4\X_7\X_8\X_9\X_{10}\\
 \epsilon & \sim & N(0,\I)
\end{IEEEeqnarray*}
The coefficients $\beta_1,\ldots,\beta_4$, are equal given the norm of the
interaction and are chosen to yield a true model \Rs of approximately .83. The
t-statistics of features in the true model range between 25 and 40. It is 
important to note that the features are tested ``back to front,'' meaning that 
features one through ten are considered last. This is a worst case ordering for 
our methods, as they are forced to test all irrelevant features first.

We first simulate a small-p environment: 2,000 observations with 350 explanatory
features. While our features are simulated independently, the maximum observed
correlation is approximately .14. While many competitor algorithms are compared
on the real data, only two are presented here for simplicity. Our goal is to
demonstrate the gains from searching complex spaces using feature selection
algorithms. Five algorithms are compared: RAI searching the interaction space,
the Lasso, random forests \citep{Brei01}, the true model, and the mean model.
The mean model merely predicts $\bar Y$ in order to bound the range of
reasonable performance between that of the true model and the mean model. Two
Lasso models are compared: the one with minimum cross-validated error (Lasso.m)
and the smallest model with cross-validated error within one standard deviation
of the minimum (Lasso.1). Since the feature space is small, it is possible
to compute the full interaction space of approximately 61,000 features. Lasso
is given this larger set, while RAI and random forests are only given the 350
marginal variables. Random forests is included such that comparison can be made
to a high-performance, off-the-shelf procedure that also constructs its
own feature space.

Figure \ref{fig:perf} compares the risk of all procedures and the size of the 
model produced by the feature selection algorithms. The risk is computed using 
squared error loss from the true mean: $\|\mu_Y-\hat Y\|_2^2$. RAI often 
outperforms the competitors even though it is provided with far less 
information, while the success of Lasso demonstrates the strength of 
correlation 
in this scenario. Even though Lasso can only accurately include the interaction 
$\X_1\X_2$, it is able to perform reasonably well in some cases. Figure 
\ref{fig:perf} resamples the data 50 times, creating cases of varying 
difficulty. Often, difficult cases are challenging for all algorithms, such 
that 
the highest risk data set is the same for all procedures. RAI performs better 
than Lasso.m on the majority of cases and almost always outperforms Lasso.1. 
The 
overlapping box plots merely demonstrates the variability in the difficulty of 
data sets.

It is also worth comparing the size of the model selected by different
procedures. The Lasso often selects a very large number of variables to account
for its inability to incorporate the correct interactions. As demonstrated 
explicitly in the concrete compressive strength example in the 
introduction, this is a general problem even when the
Lasso is provided the higher-order interactions. Using the model
identified by Lasso.1 dramatically reduces model size with a concomitant
increase in loss. Contrast this with RAI, which selects a relatively small
number of features even though its search space is conceptually infinite, as no
bounds on complexity of interactions is imposed. Furthermore, RAI necessarily
selects more than four features in order to identify the higher order terms. For
example, in order to identify the term $\X_7\X_8\X_9\X_{10}$, all four marginal
features need to be included, as well as other interactions of the form 
$\X_7\X_8$, etc. In fact, few third-order interactions are considered. 
Therefore 
the true fitted space is not much larger than that considered by the Lasso. RAI 
performs better in this case because it does not need to consider the full 
complexity of the 61,000 features in the interaction space.

\begin{figure}
\begin{center}
 \begin{subfigure}{.24\textwidth}
  \centerline{\includegraphics[width=\textwidth]{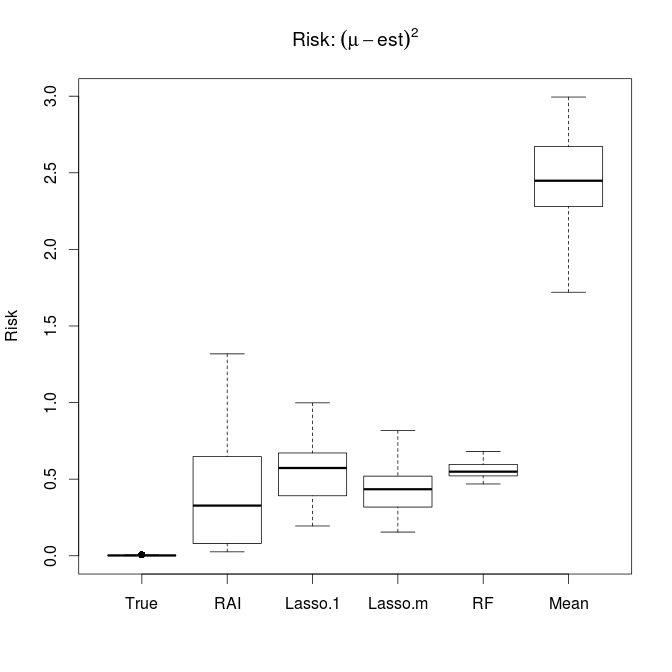}}
  \caption{Risk: $p=350$}
 \end{subfigure}
 \begin{subfigure}{.24\textwidth}
  \centerline{\includegraphics[width=\textwidth]{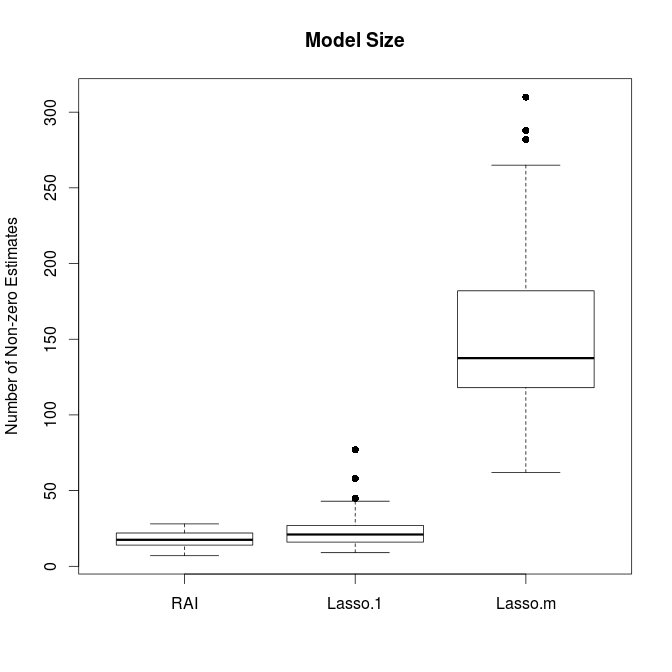}}
  \caption{Size: $p=350$}
 \end{subfigure}
  \begin{subfigure}{.24\textwidth}
  \centerline{\includegraphics[width=\textwidth]{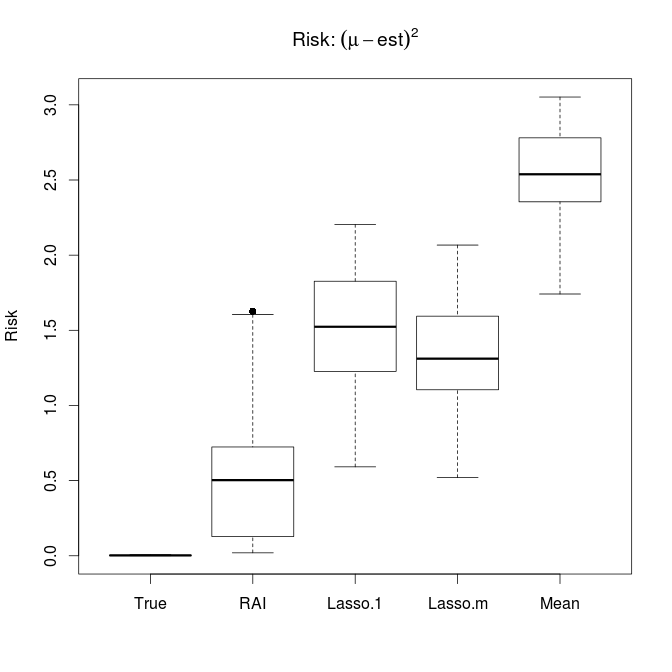}}
    \caption{Risk: $p=10,000$}
 \end{subfigure}
 \begin{subfigure}{.24\textwidth}
  \centerline{\includegraphics[width=\textwidth]{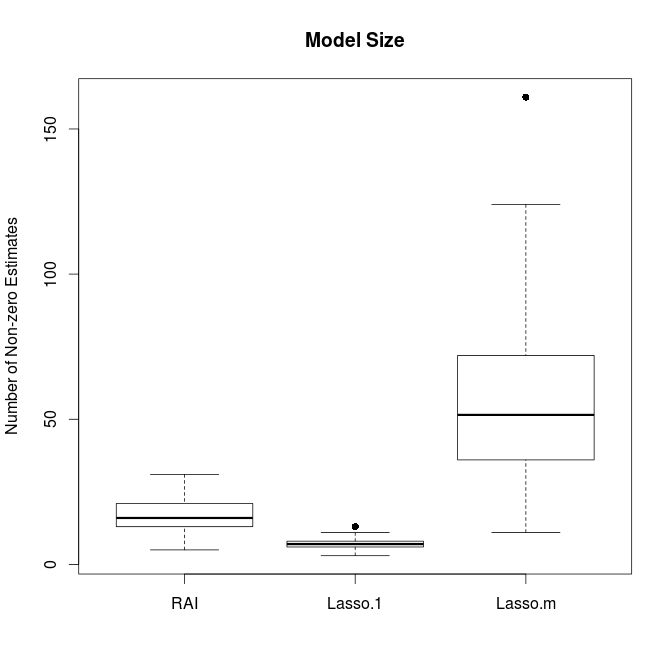}}
    \caption{Size: $p=10,000$}
 \end{subfigure}
 \caption{Risk and model size for small- and large-$p$ simulation settings.}
 \label{fig:perf}
\end{center}
\end{figure}

While our results do not focus on speed, it is worth mentioning that RAI easily
improves speed by a factor of 10-20 over the Lasso. This is notable since the
Lasso is computed using glmnet \citep{glmnet}, a highly optimized Fortran
package, while RAI is coded in R and is geared toward conceptual clarity as 
opposed to speed. Furthermore, RAI also does not have to compute the full 
interaction space.

Next, consider a comparatively large feature space: 2,000 observations with 
10,000 explanatory features. In this case, the maximum observed correlation 
between features is .177. Both RAI and the Lasso are only given the marginal 
variables because the full second-order interaction space has 50 million 
features. Traditional forward stepwise is also very time intensive to run on 
models of this size even without considering the interaction space. Therefore, 
an intelligent search procedure is required to identify signal. Random forests 
were excluded given the excessive time required to fit them off-the-shelf.

Figure \ref{fig:perf} shows the risk and model size resulting from the
algorithms fit to these data. When comparing the risk of the algorithms,
RAI always outperforms both Lasso models, often by 55-90\%. The overlapping
region in the plot merely shows the variability in the difficulty of data
cases. 

RAI only identifies 1.6 true features on average, and rarely identifies all
four. As before, this is in large part due to the number of hypotheses that 
need to be rejected in order to identify an interaction such as 
$\X_7\X_8\X_9\X_{10}$. At
least seven tests must be rejected, starting with the marginal terms, some
second- and third-order interactions, and lastly the true feature. 
That being said, significant progress toward this true feature is
made in all cases. For example, the model includes features such as $\X_7\X_8$
and $\X_8\X_9$ or $\X_7\X_9\X_{10}$.

\subsection{Real Data}
\label{sec:real-data}


One of the most common applications of microarrays is “differential expression” 
profiling: identifying mRNAs/genes whose expression level is different 
under two conditions. More complex questions have been asked recently about the 
interaction between genetics and the environment (patient behavior, etc). In 
order to demonstrate the flexibility of RAI, we address a related problem in the 
selection of higher-order interaction models of multi-factorial transcription 
data \citep{mouseGeneInteract14}.

Our application uses microarray data of bone marrow-derived macrophages 
from different inbred mouse strains \citep{mouseGeneData06} and is accessible 
under GEO accession number GSE 2973. There is an interplay between three factor 
groups: mouse genetic background, bacterial infection with 
\emph{Yersinia enterocolitica} (two different strains and a mock control), and 
an indicator for treatment with interferon-$\gamma$. It is known that C57BL/6 
mice are not susceptible to infection with \emph{Yersinia}, while BALB/c mice 
without IFN-$\gamma$ stimulation are \citep{mouseGeneResult86, 
mouseGeneResult94}. Therefore, there 
is interest in characterizing the interaction between IFN-$\gamma$ and the 
genetic background. 

\citet{mouseGeneInteract14} consider four linear model specifications which 
include various interaction terms that are presented in hierarchical order. 
Their goal is to determine which single model should be used to asses the 
interaction between genetics and environment for all genes, and introduce a 
plotting method to understand the practical significance of these decisions. 
The 
potential models are given in the following 
display, where $\bf{Y}$ is the gene expression matrix, $G$ is the genetic 
background (mouse type), IFN-$\gamma$ stimulant indicator $\Gamma$, and 
bacterial strain $I$.
\begin{IEEEeqnarray}{rCl}
 \bf{Y} & \sim & G + \Gamma + G:\Gamma + I + G:I + \Gamma:I + 
G:I:\Gamma \label{eqn:geneMod}\\
 \bf{Y} & \sim & G + \Gamma + G:\Gamma + I + G:I + \Gamma:I\label{eqn:gene-3}\\
 \bf{Y} & \sim & G + \Gamma + G:\Gamma + I\label{eqn:gene-sel}\\
 \bf{Y} & \sim & G + \Gamma + G:\Gamma\label{eqn:gene-base}
\end{IEEEeqnarray}
Higher order models were considered only for those genes that show a 
significant global effect in model (\ref{eqn:gene-base}) after an FDR 
correction.

Our goal, however, is to select the appropriate higher-order terms for each 
gene 
\emph{individually}. The data set contains 22,690 genes from 37 mice, 
and we will treat each 
regression as a separate interaction-search problem, but all regressions as a 
connected multiple-comparisons problem. All models will include the components 
$G$, $\Gamma$, and $G:\Gamma$, but the other terms in equation 
(\ref{eqn:geneMod}) will be selected for each gene $\bf{Y}_i$ individually. One 
problem, however, is how to share power across the 22,690 connected regression 
problems. Furthermore, we want to incorporate the pre-testing step into a 
valid, multi-step procedure, which the previous authors ignored. Both 
problems are easily solved within the alpha-investing framework.

First, we treat the selection of genes to consider as an instance of the 
special case in Section \ref{sec:orthog} as the global effect tests for each 
base regression model do not change based on which other hypotheses have been 
rejected. As such, we also update the testing thresholds for use with 
RH. Granted, as genes are related, these tests are not truly independent, but 
we treat them as such following \citet{mouseGeneInteract14}. Half of the 
initial alpha-wealth of .05 is spent at this step. Furthermore, instead of 
running RH to completion, we stop when the testing threshold does not increase 
considerably after completing a testing pass in order to retain some 
accumulated wealth for the second step. We note that the results of the 
combined procedure are rather insensitive to this choice of stopping criteria.

On the subset of selected genes, we run RAI on each gene at the 
bonferroni-level 
searching for higher-order interactions on top of the base-model in equation 
(\ref{eqn:gene-base}). After termination, we collect all unspent wealth and 
restart RAI on genes for which \emph{no additional features} were selected. 
When 
restarting, we could either allocate the wealth to regression models again at 
the bonferroni level, (collected $\alpha$)/(remaining regressions), or by a 
different scheme such as considering regressions sequentially and spending 
2*(previous testing level) on each. We opt for the latter in line with the 
geometric increase in testing levels used by RAI. This process is then repeated 
until all wealth is depleted. While this procedure is somewhat wasteful in that 
it ignores the previous tests, the loss in power is negligible because so 
little 
alpha-wealth was spent on these tests. More importantly, the procedure is 
easily 
implementable and valid, in that it controls mFDR over all regression problems.

\citet{mouseGeneInteract14} select and use the model in equation 
(\ref{eqn:gene-sel}) for further analysis. Notably, this does not include any 
higher-order interaction effects among the considered factor groups, opting to 
only add marginal effects for bacterial strain. Using $\alpha=.05$, our 
pretesting step yields 6,746 genes for further analysis, which is close to the 
6,446 selected using the BH procedure. Among this set, we identify expanded 
models for 958 genes, 63 of which include higher-order interactions as in 
equations (\ref{eqn:geneMod}) or (\ref{eqn:gene-3}).

In order to directly compare results, it is easiest to consider the eruption 
plot introduced in \citet{mouseGeneInteract14}, which is 
essentially two volcano plots \citep{volcano12} placed on top of each other. 
A volcano plot displays a measure of unstandardized signal 
strength such as log-fold-change against a standardized signal strength such as 
-log(p-value). In this case, we plot these measures for the coefficient 
estimate on the genetic background-environment interaction $(G:\Gamma)$ for 
each gene that passed the selection step. An eruption plot compares two 
competing models by plotting the two volcano plots together, connecting the 
points for each gene using an arrow. The base of the arrow is located at the 
observation in the base model, whereas the tip is located at the observation in 
the chosen expanded model. 
The resulting plot then includes horizontal and vertical threshold for use in 
double filtering to determine the regions of practical interest (ROI): genes 
for which the effect size is both sufficiently large and significant.

\begin{figure}
 \begin{subfigure}{.45\textwidth}
  \centerline{\includegraphics[width=\textwidth]{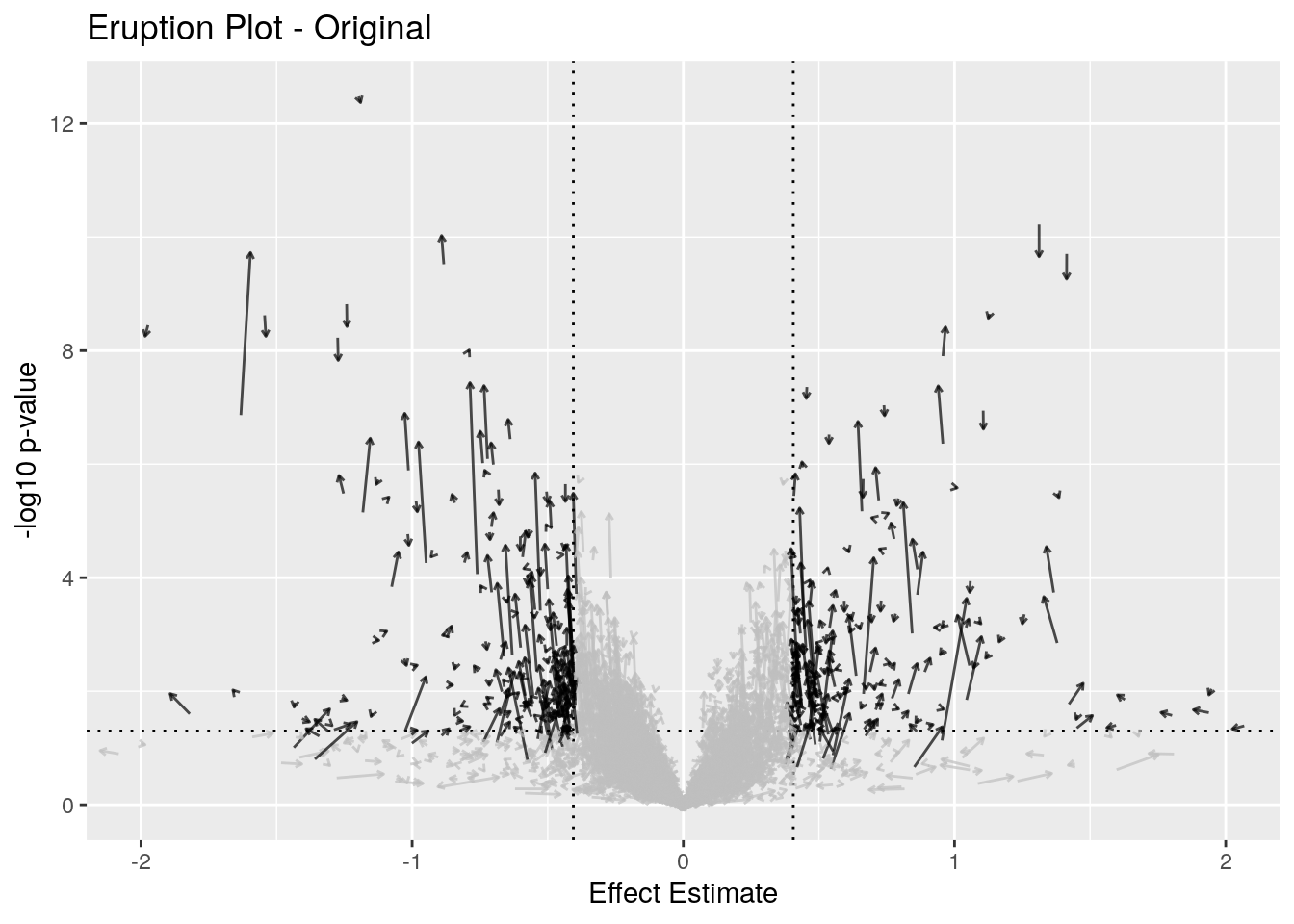}}
 \end{subfigure}
 \begin{subfigure}{.45\textwidth}
  \centerline{\includegraphics[width=\textwidth]{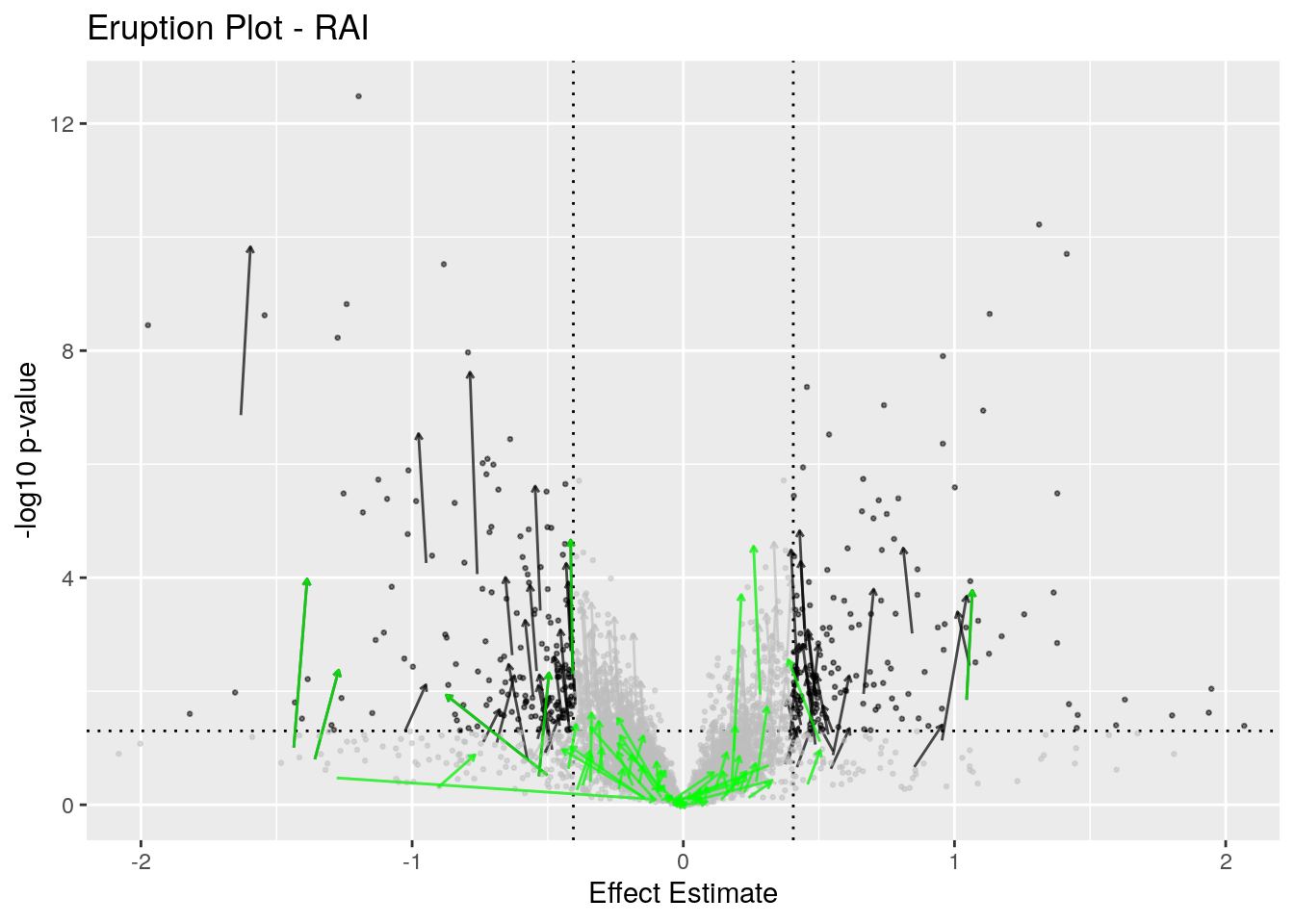}}
 \end{subfigure}
 \caption{Eruption plots for change between base and selected models. Genes 
partially contained in the regions of interest are in black. Additionally, 
arrows in green correspond to genes in which a higher-order interaction was 
identified by RAI.}
 \label{fig:eruption}
\end{figure}

The differences between the two types of analysis are visible in Figure 
\ref{fig:eruption}. Among all genes that passed the selection step, both 
analyses identify approximately 360 genes within the regions of interest, of 
which approximately 20 are different between the two analyses. Of more 
importance, however, is the relative simplicity of interpreting 
confounding or effect modification by the selected model. The refined analysis 
provided by RAI highlights only those genes for which an expanded model was 
statistically significant. We identify more significant 
interaction effects, as indicated by the much longer, vertical, green arrows 
in the ROI. Furthermore, it wrapped the entire pre-test and model 
selection analysis within a single framework to account for selection and 
multiple comparisons.

\mynotesH{Include any discussion about which genes were identified?}



\section{Discussion}
\label{sec:discussion}

This paper presents a novel algorithm, \RAIp, for approximating stepwise 
regression that has multiple types of guarantees. First,  it is proven in 
control mFDR, thereby not over-selecting features. Second, it is guaranteed 
to identify signal and approximate both the true stepwise model as well as the 
best-subset model. We demonstrated how directed search can be 
used to fit high-dimensional interaction models. The alpha-investing framework 
provides flexibility to design more complex analyses such as pre-testing 
followed by model selection, as done in Section \ref{sec:real-data}. As the 
mFDR analysis only required valid testing of a single parameter estimate, much 
of the analysis translates directly to generalized linear models (GLM). 
Future work extending the results of this paper to GLMs would allow much 
broader application as many fields in which interaction-search is important 
rely on these models.



\bibliography{../../../Bib_Stuff/Bibliography}

\begin{appendices}
 \appendix
 \label{appendix}

\section{}
\label{app:valid-tests}

To prove Theorem \ref{thm:valid-tests}, it is easier to control $\Prob(V_i = 
0|R_{[i-1]})$ in which case we have recourse to the Gaussian Correlation 
Inequality (GCI) \citep{Royen14}. A simpler version of the GCI is stated as
\begin{thm}[\citet{LatM15}]
 For any closed symmetric sets $K$, $L$ in $\bbR^d$ and any centered Gaussian 
measure $\pi$ on $\bbR^d$ we have \[\pi(K\cap L) \geq \pi(K)\pi(L).\]
\end{thm}

Using this, the proof is straight-forward if one augments the null hypothesis 
currently being tested with those of previous failed rejections. While this is 
not normally required in hypothesis testing, the sequential nature of the 
procedure requires it as does the GCI. This will be explained in more detail 
below.

\begin{proof}[Proof of Theorem \ref{thm:valid-tests}]
 Separate the set of previous rejections, $R_{[i-1]}$, into tests that 
were rejected ($R_j = 1$) and those which failed to be rejected ($R_j = 0$). 
Denote these two 
sets as $\{R_{[i-1]}\}^1$ and $\{R_{[i-1]}\}^0$, respectively. Note that both 
of 
these are subsequences of random variables which form a partition of 
$R_{[i-1]}$. 
 
 The GCI will allow us to conclude
 \[\Prob(V_i = 0|R_{[i-1]}) =  \Prob(V_i = 0|\{R_{[i-1]}\}^0, \{R_{[i-1]}\}^1) 
\geq \Prob(V_i = 0).\] 
The right-most quantity is controlled by using classical test statistics. The 
proof proceeds by demonstrating that the set  $\{R_{[i-1]}\}^1$ can be safely 
removed under normality while the effect of  $\{R_{[i-1]}\}^0$ can be 
controlled 
by the GCI.
 
 Test $i$ is conducted using a test statistic computed in a space which is 
orthogonal to all features in the current model $M$, ie, orthogonal to all 
features $j$ s.t. 
$R^{M,j}_l = 1$ for some $M$ and $l < i$. Under normality, the current test is 
independent of this set of previous tests. As such, all previous tests of these 
features can be removed in the conditioning statement.  As a feature can be 
tested multiple times, it is possible that $R^{M,j}_k = 0$ but $R^{M',j}_l = 1$ 
for some $k < l < i$. Even though $R^{M,j}_k \in \{R_{[i-1]}\}^0$, $R^{M,j}_k$ 
can be removed from the conditioning statement as feature $j$ is included in 
the 
current model. Therefore the entire set $\{R_{[i-1]}\}^1$ and a subset of 
$\{R_{[i-1]}\}^0$ can be removed. Denote the remaining set of previously failed 
rejections as $\{R_{[i-1]}\}^{0\backslash1}$.
 
  Using the GCI, the remainder of the proof is trivial and uses only elementary 
probability theory. We are assuming that the probabilities of events in 
question 
are positive as they are made under the relevant ``global'' null and are 
non-empty sets that are symmetric around 0.
 \begin{IEEEeqnarray*}{rCl}
  \Prob_{H_i}(V_i = 0|R_{[i-1]}) & =  & \Prob(V_i = 0|\{R_{[i-1]}\}^0, 
    \{R_{[i-1]}\}^1) \\
  & = & \Prob(V_i = 0|\{R_{[i-1]}\}^{0\backslash1})\\
  & = & \frac{\Prob(V_i = 0\cap \{R_{[i-1]}\}^{0\backslash1}} 
    {\Prob(\{R_{[i-1]}\}^{0\backslash1})}\\
  & \geq & \frac{\Prob(V_i = 0)\Prob(\{R_{[i-1]}\}^{0\backslash1}} 
    {\Prob(\{R_{[i-1]}\}^{0\backslash1})}\\
  & = & \Prob(V_i=0)
 \end{IEEEeqnarray*}
 
 In order to verify the assumptions of the GCI, note the following:
 \begin{enumerate}
  \item Using classical tests throughout the entire sequence of testing means 
that all events of the form $\{R_i = 0\}$ or $\{V_i = 0\}$ can be written as 
$\{|\hat t_i| \leq t^*_i\}$, where $\hat t_i$ is the test statistic for the 
$i$th test and $t^*_i$ is its critical value. Therefore, these sets are 
symmetric and closed.
  
  \item While the features being tested in $\{R_{[i-1]}\}^{0\backslash1}$ may 
have been computed in different models, the set of estimated coefficients is 
still Gaussian. The test of $H^{M_i,j}_i$ considers the estimated coefficient 
\[\hat\gamma_i = [(\X_{M_i\cup j}'\X_{M_i\cup j})^{-1}\X_{M_i\cup 
j}'Y]_{|M_i|+1} 
= \delta_j'Y\] 
for some $\delta_i \in\bbR^n$.  Note that we have assumed that 
the 
feature $j$ is appended as the last column of the data matrix $\X_{M_i\cup j}$. 
Let $\Delta = (\delta_1,\delta_2,\ldots,\delta_{i})$ be the matrix whose 
columns 
are the $\delta_i$ vectors for the set of failed tests as well as the current 
test. 
  
  \item In order to have a centered Gaussian, we need to augment the null 
hypotheses of test $i$ to incorporate the null hypotheses of the tests in 
$\{R_{[i-1]}\}^{0\backslash1}$. The original null hypotheses given in equation 
(\ref{eqn:null}) can be written as $H_i$: $\delta_i'\mu(\X) = 0$. The augmented 
null hypothesis is $\Delta'\mu(\X) = \mathbf{0}$, where $\mathbf{0}$ is a 
vector 
of zeros of length $|\{R_{[i-1]}\}^{0\backslash1}|+1$. In order to have a 
non-degenerate normal distribution, we require $n > 
|\{R_{[i-1]}\}^{0\backslash1}|+1$. $\hfill\qedhere$
 \end{enumerate}
\end{proof}

\section{}
\label{app:performance}

Before proving Theorem \ref{thm:performance}, we derive the p-value thresholds 
used to search 
for a given increase in \Rs. As before, we assume that $\Var(Y) = \Var(X_i) = 1$ 
and $\bar Y = \bar X_i = 0$, $\forall i$. For additional simplicity, the 
normalizations involved in $\Var(Y) = 1$ etc ignore degree of freedom 
adjustments. Therefore, we assume $Y'Y = X_i'X_i = n$ instead of $n-1$ etc. 


\subsection{Identifying p-value Thresholds}

RAI and RAI$^+$ search for features that result in an increase of $r^s$ in \Rs 
for the current model. It is well known that the \Rs in a simple regression 
model between $Y$ and $X_i$ is just the squared correlation between $Y$ and 
$X_i$, $r_{Y,i}^2$. Similarly, the increase in \Rs from adding a feature $j$ to 
the model $M$ is the squared partial correlation. We remind the reader of this 
in the lemma below, in slightly greater generality which we will need later.

We write $\mRs(S.M)$ to be the contribution to \Rs of the features in $S.M$. 
Stated differently, $\mRs(S.M)$ is the improvement in \Rs by adding set $S$ to 
set $M$.
\begin{lemma}
 \label{lem:r2-sep}
 Given subsets of features $M$ and $S$,
 \[\mRs(M \cup S) = \mRs(M) + \mRs(S.M)\]
\end{lemma}
\begin{proof}
 Let $\X_{M,S.M} = [\X_M,\X_{S.M}]$ and $M = [M,S.M]$.
 \begin{IEEEeqnarray*}{rCl}
  \mRs(M \cup S) & = & \mRs(M \cup S.M)\\
  & = & n^{-1}Y'P_MY + n^{-1}Y'P_{S.M}Y\\
  & = &\mRs(M) + \mRs(S.M)
 \end{IEEEeqnarray*}
 The first line follows because the prediction space did not change and the
 second line follows because $\X_{M,S.M}^T\X_{M,S.M}$ is block
 diagonal. Note that this could also be taken as a definition of $\mRs(S.M)$.
\end{proof}

In the above lemma, replacing $S$ with a single column $j \notin M$, shows the 
marginal improvement in \Rs to be the squared partial correlation between $Y$ 
and $\X_{j.M}$, denoted $r_{y,j.M}$. Therefore, RAI computes and compares 
squared partial correlations to the current threshold $r^s$. A feature $j$ is 
added to the current model $M$ when $r_{y,j.M}^2 > r^s$. As the t-statistic for 
testing $H_i^{M,j}$ can be written as a function of this partial correlation, 
the comparison RAI is making can be written in terms of p-values or test 
statistics.

\subsection{Bounding the Performance of the Selected Set}

We begin by generalizing the standard greedy proof of \citep{NemWF78} to 
approximately submodular functions. The proof is also similar to that of 
\citet{DasK11}, but allows $l > k$. Afterward we will discuss what needs to be 
changed for RAI$^+$

\begin{prop} If $M_l$ is selected by stepwise regression, then
 \label{prop:step}
 \begin{IEEEeqnarray*}{rCl}
  \mRs(M_l) & \geq & (1 - e^{-l\gamma/k})\mRs(M_k^*)
 \end{IEEEeqnarray*}
\end{prop}

Proving this proposition requires a bound on the difference between the \Rs of 
adding a set of features and the sum of the changes in \Rs by adding the 
features one at a time. Such a bound is provided by the submodularity ratio 
$\gamma$.

\begin{lemma}
 \label{lem:Rsbnd}
 For simplicity let $S\cap M = \emptyset$, or define $\tilde S = S\backslash M$.
 Then,
 \begin{IEEEeqnarray*}{rCl}
  \mRs(S.M) & \leq & \frac{\sum_{x\in S\backslash M} \mRs(M\cup
   \{x\})- \mRs(M)}{\gamma(M,|S|)}
 \end{IEEEeqnarray*}
\end{lemma}
\begin{proof}
 \begin{IEEEeqnarray*}{rCl}
  \mRs(S.M) & = & (r^M_S)'(C^M_S)^{-1}(r^M_S)\\
  & \leq & \frac{(r^M_S)'(r^M_S)}{\gamma(M,|S|)}\\
  & = & \frac{\sum_{x\in S.M} \mRs(\{x\})}{\gamma(M,|S|)},
 \end{IEEEeqnarray*}
 where the inequality follows by the definition of $\gamma(M,|S|)$. Since each
 element in $r^M_S$ is a correlation, squaring this gives the \Rs from the
 simple regression of $Y$ on $\x.M$, giving the final equality. The lemma
 just rewrites the result of the projection off of $M$ as a difference in
 observed \Rs.
\end{proof}

\begin{proof}[Proof of Proposition \ref{prop:step}]
 \begin{IEEEeqnarray*}{rCl+s}
  \mRs(M_k^*) & \leq & \mRs(M_i \cup M_k^*) & by monotonicity\\
  & = & \mRs(M_i) + \mRs(M_k^*/M_i) & Lemma \ref{lem:r2-sep}\\
  & \leq & \mRs(M_i) + \frac{\sum_{x\in M_k^*\backslash M_i}
   \mRs(M_i \cup \{x\}) - \mRs(M_i)}{\gamma_{M_i,|M_k^*\backslash M_i|}} & Lemma
  \ref{lem:Rsbnd}\\
  & \leq & \mRs(M_i) + \frac{k}{\gamma_{M_i,|M_k^*\backslash
    M_i|}}\max_{x\in M_k^* \backslash M_i} \mRs(\{x\})& sum less than k*max\\
  & \leq & \mRs(M_i) + \frac{k}{\gamma_{M_i,|M_k^*\backslash M_i|}}
  (\mRs(M_{i+1}) - \mRs(M_i)) & by greedy algorithm
 \end{IEEEeqnarray*}
 
 Increasing the size of the set $S$ by inclusion and increasing k can only 
decrease $\gamma(S,k)$. Therefore, $\gamma(M_i,|M_k^*\backslash M_i|) \geq 
\gamma(M_l,|M_k^*|) \geq \gamma$. Rearranging the final line above, 
dividing by $k/\gamma$ and adding $(1- \gamma/k)\mRs(M_k^*)$ to both sides 
yields
 \begin{IEEEeqnarray*}{rCl+s}
  \mRs(M_k^*) - \mRs(M_{i+1}) & \leq & (1 - \gamma/k)(\mRs(M_k^*) - \mRs(M_i))
  \IEEEyesnumber\label{eqn:ind.step}\\
  & \leq & (1 - \gamma/k)^{i+1}(\mRs(M_k^*) - \mRs(M_0)) &repeatedly apply \ref{eqn:ind.step}\\
  & = & (1 - \gamma/k)^{i+1}\mRs(M_k^*) & $\mRs(M_0) = 0$\\
  & \leq & e^{-(i+1)\gamma/k}\mRs(M_k^*) & Taylor approximation\\
  \Rightarrow \mRs(M_l) & \geq & (1 - e^{-l\gamma/k})\mRs(M_k^*) & set $l = i+1$ \qedhere
 \end{IEEEeqnarray*}
\end{proof}

The above proof cannot be used to control RAI because we are not guaranteed to 
include the feature which yields the maximum increase in \Rs. Furthermore, RAI 
is not even guaranteed to observe $\mRs(M_i \cup \{x\}) - \mRs(M_i)$, because 
the previous tests of $x\in M_k^*\backslash M_i$ may have occurred in a 
different model $M_{i-\iota}$, for $1 \leq \iota \leq i$. Overcoming these 
difficulties requires the modifications in RAI$^+$.

For RAI$^+$, we must be precise about the distinction between the size of 
the current model and the index of the testing pass. As before, the current 
model of size $i$ is denoted $M_i$. Let the current testing pass be $s$ and  
$M^{s-1}$ be the model used at the end of the previous testing pass $s-1$, where 
the pass index is given as a superscript to differentiate it from the subscript 
of step number. Note that $M^{s-1} = M_{i-\iota}$ for some $\iota \leq i$. 
Therefore $\iota$ is the number of features added in this pass before the 
current test. Recall that, by the definition of \RAIp, all features are tested 
in $M^{s-1}$ and found to yield improvements in \Rs less than $r^{s-1}$.

We desire a similar bound as in Lemma \ref{lem:Rsbnd} which applies to models 
selected by RAI$^+$. Such a bound is given in the following lemma.
\begin{lemma}
 \label{lem:rai.bound}
 If $M_{i+1}$ and $M_i$ are models chosen by RAI$^+$, then the marginal 
improvement of adding the set of features in $M_k^*$ to the model $M_i$ can be 
bounded as
 \begin{IEEEeqnarray*}{rCl}
  \mRs(M_k^*.M_i) & \leq & c(\mRs(M_{i+1} - \mRs(M_i)),
 \end{IEEEeqnarray*}
 where $c = \left( \frac{\iota+k}{\gamma r} - \iota \right)$ and $\iota = |M_i \backslash M^{s-1}|.$
\end{lemma}

\begin{proof}
 The bound is the result of rearranging the submodularity ratio while 
accounting for the fact that marginal improvements in \Rs can only be bounded in 
model $M^{s-1}$. In what follows, all statistics are computed as additions to 
model $M^{s-1}$. For example, $r_{y,j}^2$ is really $r_{y,j.M^{s-1}}^2$. Without 
loss of generality, let the new features added in the current pass be $M_i 
\backslash M^{s-1} = \{1,\ldots,\iota\}$.  Similarly, let $M_k^*\backslash M_i = 
\{\iota+1,\ldots,\iota+|M_k^*\backslash M_i|\}$. By definition, the 
submodularity ratio yields
 \begin{IEEEeqnarray*}{rCl+s}
  \gamma & \leq & \frac{r_{y,1}^2 + \ldots + r_{y,\iota+|M_k^*\backslash M_i|}^2}
  {\mRs(M_i \backslash M^{s-1} \cup M_k^*)}\\
  & = & \frac{r_{y,1}^2 + \ldots + r_{y,\iota+|M_k^*\backslash M_i|}^2}
  {r_{y,1}^2 + r_{y,2.1}^2 + \ldots + r_{y,\iota.\{1,\ldots,\iota-1\}}^2 + 
\mRs(M_k^*.M_i)}.
  & Lemma \ref{lem:r2-sep}
 \end{IEEEeqnarray*}
 
 All of the objects in the numerator were below the bound $r^{s-1}$ because 
they were tested before the end of the previous pass. Similarly, all of the 
squared partial correlations in the denominator were greater than the threshold 
$r^s$, because they were added to the RAI$^+$ model during the current pass. 
Using these bounds yields
 \begin{IEEEeqnarray*}{rCl}
  \mRs(M_k^*.M_i) & \leq & \frac{(\iota+|M_k^*\backslash M_i|)r^{s-1}}{\gamma} - \iota r^s\\
  & = & \left( \frac{\iota+|M_k^*\backslash M_i|}{\gamma r} - \iota \right) r^s\\
  & \leq & \left( \frac{\iota+|M_k^*\backslash M_i|}{\gamma r} - \iota \right) 
  (\mRs(M_{i+1}) - \mRs(M_i)) \IEEEyesnumber\label{eqn:lem3.bnd}.
 \end{IEEEeqnarray*}
 Where the last line follows when the next feature is added during pass $s$. If 
not, then we know $\iota = 0$, the elements in the numerator are bounded by 
$r^s$ not $r^{s-1}$, and we can pull out $r^{s+1}$ instead of $r^s$. This yields 
the same bound as in equation \ref{eqn:lem3.bnd}. In the worst case, no features 
in $M_k^*$ are in the current model. Setting $|M_k^*\backslash M_i| = k$ yields 
the statement in the lemma.
\end{proof}

Using this lemma, the main theorem can be proven.
\begin{proof}[Proof of Theorem \ref{thm:performance}]
Replace $k/\gamma$ in 
\begin{IEEEeqnarray*}{rCl}
\mRs(M_k^*) & \leq & \mRs(M_i) + 
\frac{k}{\gamma_{M_i,|M_k^*\backslash M_i|}} (\mRs(M_{i+1}) - \mRs(M_i)) 
\end{IEEEeqnarray*}
by $c$ given in Lemma 
\ref{lem:rai.bound}. The proof proceeds as before.
\end{proof}
\end{appendices}

\end{document}